\documentclass[twocolumn,twoside]{IEEEtran}
\IEEEoverridecommandlockouts
\usepackage{comment}

\usepackage{dsfont}
\usepackage{cite}
\usepackage{stfloats}

\usepackage{psfrag}
\usepackage[mathscr]{euscript}
\usepackage{acronym}  
\usepackage{algorithm}
\usepackage[noend]{algpseudocode}
\usepackage{booktabs}

\usepackage{bbm}
\usepackage{listings}
\usepackage{cite}
\usepackage{url}
\usepackage{xcolor}
\usepackage{graphicx}
\usepackage{epstopdf}
\usepackage[cmex10]{amsmath}
\usepackage{amssymb}
\usepackage{amsthm}
\usepackage{mathrsfs}
\usepackage{verbatim}
\usepackage[noend]{algpseudocode}

\usepackage{setspace}
\usepackage{bm}%
\usepackage{stfloats}
\usepackage{enumerate}

\theoremstyle{plain}
\usepackage{multirow}
\newtheorem{theorem}{Theorem}
\newtheorem{remark}{Remark}
\newtheorem{lemma}[]{Lemma}
\newtheorem{corollary}{Corollary}

\usepackage{array}
\usepackage[caption=false,font=footnotesize,labelfont=normalfont,textfont=normalfont]{subfig}

\makeatletter
\def\BState{\State\hskip-\ALG@thistlm}
\makeatother

\usepackage{float}
\usepackage{balance}

\newlength \figwidth
\usepackage[acronym]{glossaries}

\makeglossaries
\newacronym{NTN}{NTN}{Non-terrestrial networks}
\newacronym{IoT}{IoT}{Internet of Things}
\newacronym{6G}{6G}{sixth-generation}
\newacronym{EH}{EH}{Energy Harvesting}
\newacronym{AoI}{AoI}{Age of Information}
\newacronym{LEO}{LEO}{Low Earth Orbit}
\newacronym{SG}{SG}{stochastic geometry}
\newacronym{SNR}{SNR}{signal-to-noise ratio}
\newacronym{IID}{\emph{i.i.d.}}{independently and identically distributed}
\newacronym{LoS}{LoS}{line-of-sight}

\begin{document}

\title{Age of Information in Non-Terrestrial Networks with Energy Harvesting}

\author{
Fangming~Zhao,      
Nikolaos Pappas,
Shi Jin, 
and~Howard~H.~Yang

\thanks{F.~Zhao and H.~H.~Yang are with the ZJU-UIUC Institute, Zhejiang University, Haining 314400, China (e-mail: fangming.23@intl.zju.edu.cn; haoyang@intl.zju.edu.cn).
}


\thanks{ Nikolaos Pappas is with the Department of Computer and Information Science, Linkoping University, Linkoping 58183, Sweden (e-mail: nikolaos.pappas@liu.se).}

\thanks{S. Jin is with the National Mobile Communications Research Laboratory, Southeast University, Nanjing 210096, China (e-mail: jinshi@seu.edu.cn).}

}

                    

\maketitle

\begin{abstract}
We analyze the timeliness of status-update delivery in a low Earth orbit (LEO) satellite-assisted energy-harvesting Internet of Things network using the Age of Information (AoI) metric. A ground source harvests ambient energy and sends status updates to a remote destination
through LEO satellites. Because of satellite mobility, source-to-satellite connectivity alternates between on and off periods whose durations depend on the satellite-ground geometry. The source does not know the connectivity state a priori and therefore employs
a probe-before-transmission mechanism: it first expends one energy unit to sense satellite availability and transmits an update only after a successful probe. We combine spherical stochastic geometry with semi-Markov analysis to characterize the coupled evolution of
satellite connectivity and the source energy buffer, and derive an analytical expression for the time-average AoI. We then develop a lower-complexity approximation by replacing the instantaneous connectivity state in the energy process with the long-term on-state
probability. The resulting approximation is accurate when the energy constraint is weak or satellite connectivity is highly intermittent. Numerical results show that probing can substantially reduce AoI relative to blind transmission by preventing energy expenditure during off periods, particularly under sparse satellite deployment, stringent decoding requirements, or limited energy harvesting.
\end{abstract}



\section{Introduction}\label{sec:intro}
\gls{NTN} are expected to extend connectivity to infrastructure-limited regions by leveraging \gls{LEO} satellites. Meanwhile, ground \gls{IoT} devices deployed for remote sensing in these environments are typically energy-constrained, since battery replacement and wired recharging are costly. \gls{EH} technology provides a promising solution by enabling self-sustained operation through renewable or ambient energy sources, such as solar and radio-frequency energy\cite{EHSurvey:JSAC}. LEO satellite-assisted EH-IoT networks are an emerging research direction for sustainable remote IoT connectivity.


\par
The integration of \gls{LEO} and EH-IoT, however, introduces several performance-evaluation challenges. First, satellite connectivity is inherently intermittent due to orbital dynamics and limited visibility windows, causing the channel to alternate between available and unavailable states, and the distribution of available periods is governed by the satellite trajectory. Second, with \gls{EH}, transmission opportunities are further constrained by the stochastic energy state of the device, and an update can be attempted only when sufficient energy has been accumulated. Third, an \gls{IoT} device cannot know the channel state in advance. Therefore, it first probes the channel and transmits data only when the channel is available. Since probing consumes different energy in on and off states, the energy process becomes state-dependent and correlated with the satellite visibility process. This \textit{channel-energy availability coupling} invalidates a simple one-dimensional energy-state Markov model and motivates a joint semi-Markov analysis.


\par
Furthermore, since the data center relies on information collected from remote \gls{IoT} devices for monitoring and decision-making, conventional metrics such as throughput or delay are insufficient. A metric is needed that captures how fresh the delivered information is under intermittent connectivity and energy-limited transmissions. This motivates the use of the \gls{AoI} at the receiver. \gls{AoI} provides such a metric by measuring the time elapsed since the most recently received update was generated \cite{JSACAoIsurvey, pappas2023age}. It has been widely used to quantify information timeliness and guide the design of update policies \cite{kadota2018optimizing}. In LEO satellite-assisted EH-IoT networks, the \gls{AoI} analysis is particularly challenging because the update process is governed by the joint evolution of satellite visibility, energy accumulation, probing decisions. Motivated by these observations, this work analyzes the time-average \gls{AoI} of LEO satellite-assisted EH-IoT network. The key analytical challenge lies in characterizing the coupling between stochastic channel availability and energy availability. 

\subsection{Related Work}
The timeliness of information in non-terrestrial networks (NTNs) has recently attracted increasing attention. Early analytical studies have investigated latency and AoI in multi-hop satellite networks, where satellite relays are modeled through queueing systems to characterize end-to-end freshness performance \cite{Petar:ICC2020}. More recent works have further considered the unique characteristics of LEO satellite communications. \cite{TMC2025:AoINTN:Mhop} studied AoI variation in LEO satellite-terrestrial uplink transmissions under time-varying link conditions, while \cite{TVT2025:AoINTN} analyzed information freshness in multi-hop satellite IoT systems with reliability mechanisms such as ARQ/HARQ. Beyond conventional freshness metrics, \cite{ErfanNikos:NTNAoI} considered semantics-aware unified terrestrial non-terrestrial networks, where timeliness is jointly considered with information relevance and utility. However, their analytical models typically characterize the satellite link through simplified service processes, prescribed connection patterns, or abstract time-varying link states. As a result, the spatial randomness of LEO satellite deployments and satellite visibility is often not explicitly captured.

Spherical \gls{SG} has been recognized as an effective tool for characterizing large-scale NTN deployments\cite{wang2025modeling}, satellite availability, coverage probability, and link-level performance. For instance, \cite{al2021analytic} developed a tractable analytical framework for modeling the downlink coverage probability of dense satellite networks by incorporating satellite-to-ground path loss and LoS probability. \cite{talgat2021stochastic} developed a \gls{SG} analytical framework for LEO satellite systems with ground gateways, and highlighting the potential of LEO satellites to enhance connectivity in remote areas. \cite{JungTCOM2022ShadowedRicianSG} analyzed LEO satellite downlink systems under shadowed-Rician fading by modeling satellites as a homogeneous BPP on the sphere, and derived outage probability. \cite{SongTWC2023CooperativeSAT} extended \gls{SG}-based modeling to cooperative satellite-aerial-terrestrial systems, providing analytical tools for evaluating coverage rate in integrated space-air-ground architectures. 
\cite{park2022tractable} proposed a tractable approach to downlink coverage analysis in satellite networks by incorporating spatial randomness, path loss, and fading effects. \cite{AlDopplerCL} analyzed the Doppler shift distribution in satellite constellations.
\cite{IoTJ2025:NTNIoTSG} study IoT-over-LEO satellite systems under more realistic operational constraints, such as finite terrestrial regions, limited satellite coverage. \cite{tang2025sinrcoverageleo} analyze SINR coverage in LEO satellite networks by modeling satellite locations as a strong ball-regulated point process on the sphere, which captures the locally repulsive property induced by inter-satellite safety distances and yields tractable lower bounds. The above \gls{SG}-based NTN studies provide powerful statistical tools for characterizing satellite visibility and coverage, but their performance metrics are mainly coverage probability, rate, and availability, rather than information freshness. 

Recent studies have started to combine \gls{SG} with AoI analysis in NTNs. For instance, \cite{yanwuGlobecom} modeled the service process between LEO satellites and a source node as an on-off process and derived a closed-form expression for the time-average AoI by leveraging \gls{SG}. \cite{TWC2020:NTNPAoI} analyze the average PAoI in LEO satellite-enabled IoT networks, where ground IoT nodes and LEO satellites are modeled as independent PPPs. These studies provide important insights into of LEO network performance, but the coupling between energy availability and satellite channel availability remains insufficiently characterized.

In summary, most existing spherical SG-based studies of satellite networks rely on spatial snapshots for instantaneous coverage probability or rate analysis. However, they generally overlook the temporal switching of satellite visibility and the energy-buffer dynamics of EH-IoT devices, which are crucial for characterizing AoI performance. Furthermore, in remote IoT scenarios, a device often cannot know satellite channel availability in advance and must actively probe the channel before transmission\cite{li2026ProbeAoI}. Since probing in unavailable states and successful transmission in available states consume different amounts of energy, the energy process and the satellite-visibility process become coupled. This motivates an AoI-oriented analysis that jointly accounts for intermittent satellite connectivity and state-dependent energy consumption.

\subsection{Contributions}
The contributions of this work are summarized as follows:
\begin{itemize}
\item We present an analytical study of the \gls{AoI} in LEO satellite-assisted EH-IoT networks, focusing on a long-distance scenario where a ground node updates its status information to a remote destination through LEO satellites. We transform the spatial distribution and mobility of satellites into a temporal on–off process to characterize the intermittent availability of satellite service. When the node lacks knowledge of the on–off state, we introduce a probe-before-transmission scheme, enabling the source to sense link availability prior to each transmission and thereby reducing energy wasted during off periods.
\item To capture both the spatial randomness of satellite positions and the temporal dynamics of energy harvesting, we integrate a spherical Poisson point process with semi-Markov analysis, yielding the analytical expression for the time-average \gls{AoI}. This unified framework effectively bridges spatial and temporal variability, providing a tractable tool for evaluating timeliness in LEO satellite-assisted EH-IoT networks. We then develop a lower-complexity approximation by replacing the instantaneous connectivity state in the energy process with the long-term on-state probability.
\item Through numerical analysis, we further extract several insights into the \gls{AoI} behavior of LEO satellite-assisted EH-IoT. Specifically, we observe that ($i$) under highly intermittent connectivity, the probe-before-transmission strategy provides superior timeliness compared to the blind-transmission scheme; ($ii$) the optimal \gls{AoI} performance is achieved when the average energy consumption rate is equal to or greater than the harvesting rate; and ($iii$) in dense satellite constellations, \gls{AoI} becomes insensitive mainly to the decoding threshold, being primarily governed by the update frequency and satellite density.
\end{itemize}

\section{System Model}\label{sec:sysmod}

In this section, we detail the network configuration, the energy-harvesting processes and utilization strategy, and the performance metric. 

\begin{figure}[t] 
\centering{}
{\includegraphics[width=\figwidth]{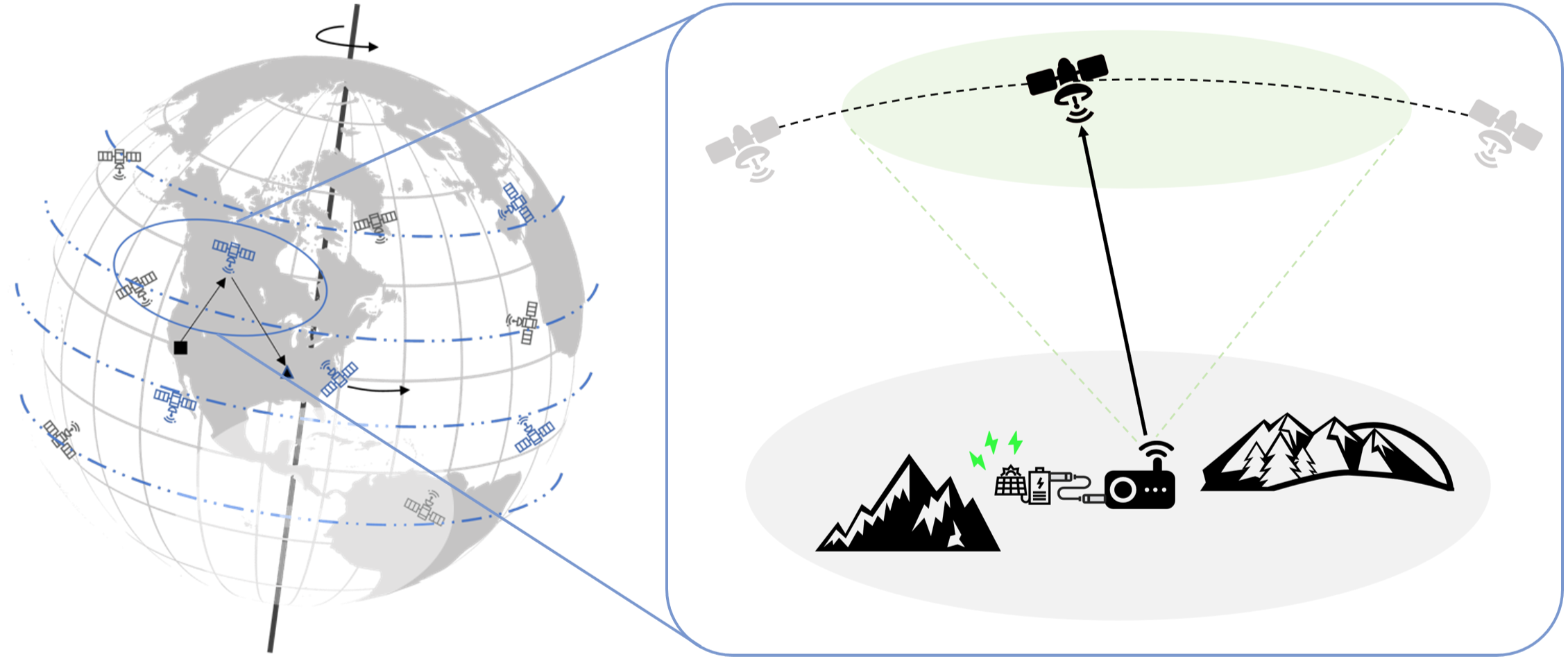}}
\caption{A snapshot of the scenario under consideration, with a constellation of LEO satellites deployed on a sphere and connecting source nodes.}
\label{fig:Sat}
\vspace{-0.3cm}
\end{figure}

\subsection{Network Configuration}
We consider an NTN consisting of a constellation of \gls{LEO} satellites deployed at the same altitude $h$. We assume the positions of the satellites to follow a homogeneous Poisson point process (PPP) \cite{al2021analytic, park2022tractable} of intensity $\lambda$ on a sphere of radius $R_E+h$,{\footnote{In practice, spatial density of the \gls{LEO} satellites can be estimated as $\lambda=\frac{N_S}{4\pi(R_E+h)^2}$, where $N_S$ is the number of satellites deployed\cite{AlDopplerCL}.}} where $R_E$ denotes the Earth's radius. 
We focus on a source node on the ground that needs to send a sequence of information packets, each containing its latest status information, to a destination node. We consider the scenario in which the source is outside the coverage area of a terrestrial network. Hence, communication needs to take place via the \gls{NTN}, through satellite nodes. An illustrative example of this setup is provided in Fig.~\ref{fig:Sat}.

We assume the source node harvests energy from the ambient environment and uses only the harvested energy for packet transmission. 
Specifically, we model the arrival of energy units as a Poisson process (in time) with rate $\xi$; and the source node stores each incoming energy unit into an energy buffer with capacity $B$.

\subsection{Packet Transmissions}
We assume the source node employs the \textit{generate-at-will} policy for status updates, in which a packet is transmitted immediately upon generation.  
Specifically, we model the interval between consecutive update attempts of two consecutive status updates as \gls{IID}, following an exponential distribution with rate $\mu$. 
We further assume the source node sends out information packets (to a relay satellite) at a fixed transmit power $P_{\mathrm{tx}}$, which consumes $N$ units of energy.
The signal propagation is subject to path loss that obeys a power law with path loss exponent $\alpha$, and reception experiences white Gaussian noise with variance $\sigma^2$. 

When the source transmits, if the \gls{SNR} received at the satellite exceeds a decoding threshold $\theta$, we consider the source node to be within the \gls{LEO}'s coverage. 
Out of the satellites from which it receives coverage, the source node then connects to at most one, namely the one providing the highest received power. 
If the source node is connected to a satellite, the corresponding information packet can be successfully delivered to the satellite after experiencing a constant propagation delay $D$.

Due to the orbital motion of \gls{LEO} satellites, a satellite currently serving the source node would eventually roam away. And before the next available satellite becomes visible, the source node is temporarily \textit{out of coverage}. 
We assume that the source node has no information about the \gls{LEO}'s operation pattern.
As such, to prevent the source node from sending information packets during the off coverage period, which is a waste of energy, we stipulate the transmitter to spend one unit of its harvested energy to probe connectivity from the \gls{LEO} before each transmission, and only transmits by confirming it is \textit{in coverage}.
To this end, the source node shall accumulate at least $N+1$ energy units before initiating a transmission/updating attempt\footnote{In this paper, we focus on probing-based timeliness enhancement over stochastic satellite channels and the resulting coupling between energy availability and channel availability. More detailed fading models and multi-user resource contention are beyond the scope of this work. Nevertheless, the framework can be extended by incorporating probabilistic link availability under, e.g., Nakagami-$m$ fading, and by introducing an on-state access success probability to capture resource contention.}. The successful update further experiences a deterministic system time of $3D$, which accounts for the probe request, probe response, and payload transmission. To reduce handover signaling and energy expenditure, we consider a single-satellite association policy. Once the source acquires a serving satellite, it remains associated with this satellite during the corresponding contact period and does not switch to other satellites before the current contact terminates.

\begin{figure}[t!] 
  \centering{}
    {\includegraphics[width=\figwidth]{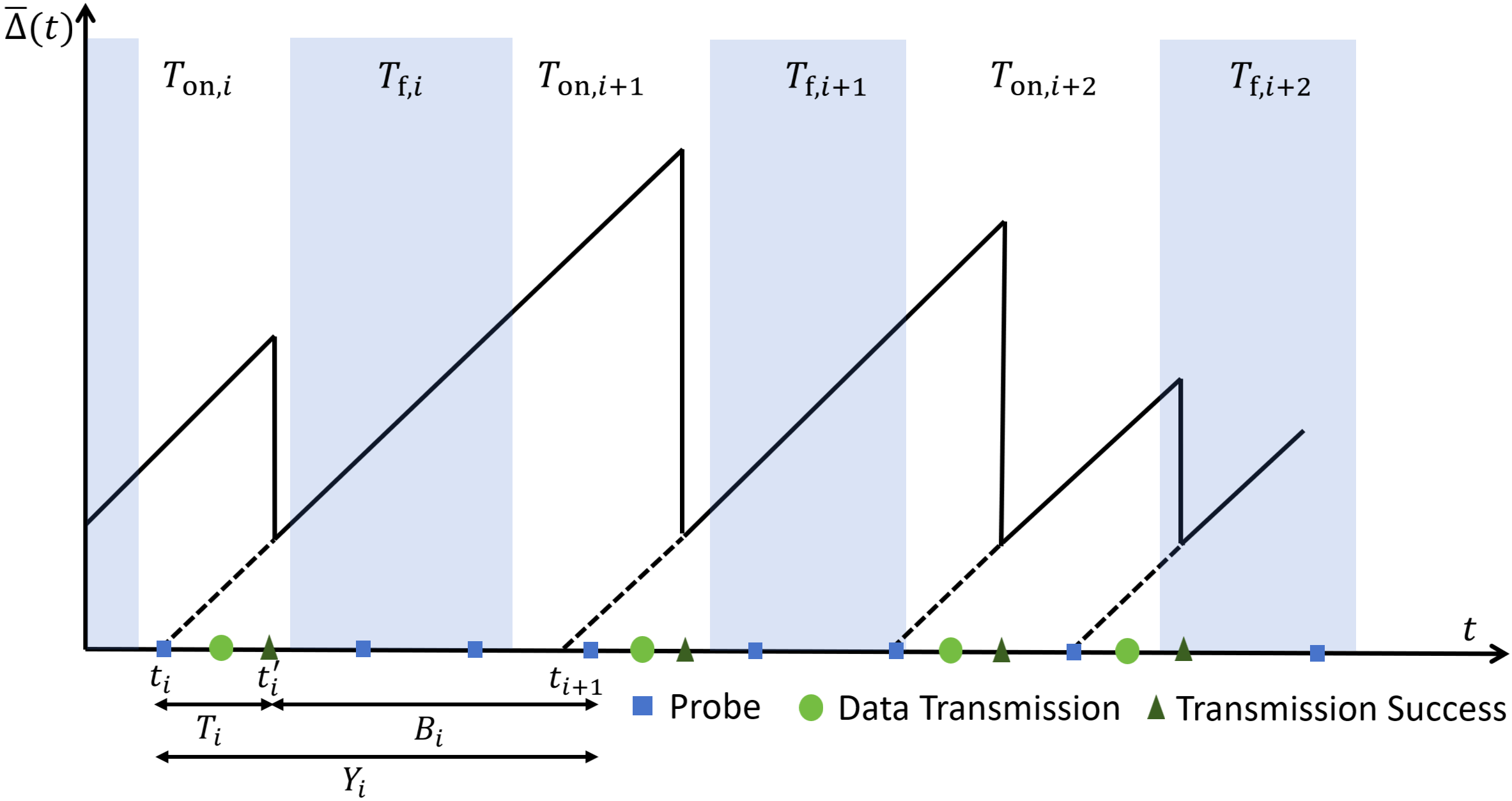}}
  \caption{Example of AoI evolution under on- and off-service periods, which is relevant to intermittent NTN connectivity.}
  \label{fig:AoI}
  \vspace{-0.5cm}
\end{figure}

\subsection{Performance Metric}

In this work, we use the \gls{AoI} metric to evaluate the system performance. 
\gls{AoI} quantifies the \emph{freshness} of information received at a destination node.
Formally, the AoI of the source-destination pair evolves as follows:
\begin{align} \label{equ:AoI}
    \Delta (t) = t - G(t),
\end{align}
where $G(t)$ is the timestamp at which the latest update received by the destination at time $t$ was generated at the source. 

The inherent dynamics of the \gls{NTN} cause the availability of satellite connectivity to be intermittent. Therefore, updates from the source node can only be successfully delivered while such connectivity is available. As a result, the evolution of the \gls{AoI} in an \gls{NTN} follows a trajectory as illustrated in Fig.~\ref{fig:AoI}, where shaded areas correspond to intervals without connectivity, during which any newly generated updates are lost. To capture the overall timeliness in the delivery of status updates through an \gls{NTN}, we then define the time-average \gls{AoI} as follows:
\begin{align}
    \bar{\Delta} = \lim_{T \to \infty }{\frac{1}{T}\int_{0}^{T} \Delta(t)}.
\end{align}

In the sequel, we derive the analytical expression for $\bar{\Delta}$, and we shed light on how it is affected by the key \gls{NTN} system-level parameters.

\section{Analysis of \gls{AoI} in \gls{NTN}}\label{sec:age_analysis}

We begin by characterizing the distribution of the intervals during which the source node is connected/disconnected to/from the \gls{NTN}, which we refer to as the \emph{on}- and \emph{off-service} periods. We then derive the analytical expression for the time-average \gls{AoI} and discuss special cases and approximations to provide additional insights. 

\subsection{Distribution of the On-Off Service Periods}

 
\begin{figure*}[!t]
    \centering
    \subfloat[Satellite motion from the node's perspective.]{
        \includegraphics[width=0.33\linewidth]{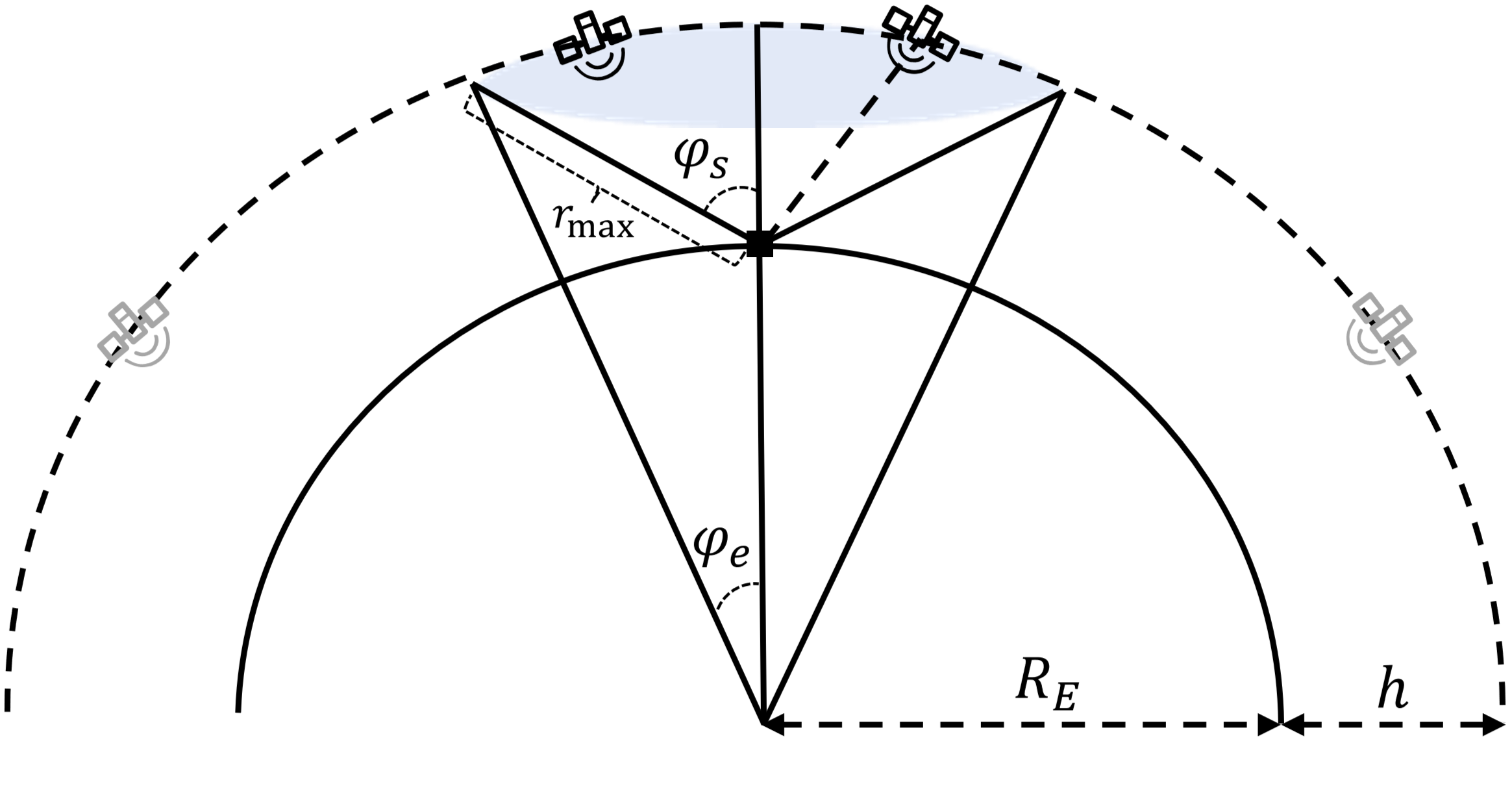}
        \label{fig:sketch}
    }\hfil
    \subfloat[Node motion from the satellite's perspective.]{
    
        \includegraphics[width=0.33\linewidth]{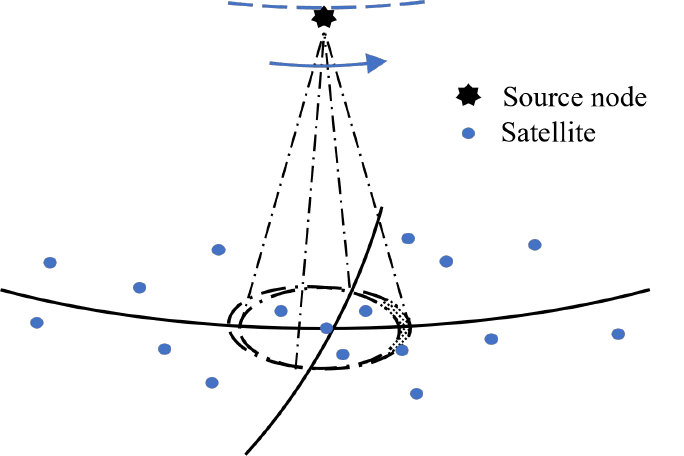}
        \label{fig:user}
    }
    \caption{Illustration of the considered NTN traffic link: (a) Side view of the NTN, with a shaded area containing satellites whose distance from the source node is smaller than $r_\mathrm{max}$; and (b) NTN as seen from the standpoint of the source node, with a shaded area indicating the relative movement of the satellites.}
    \label{fig:ntn_views}
    \vspace{-0.4cm}
\end{figure*}


At time $t$, let $r(t)$ denote the distance between the source node and the closest satellite. The corresponding SNR can be expressed as
\begin{align} \label{equ:SNR_exprss}
    \text{SNR}(t) = \frac{P_{\mathrm{tx}} \, r(t)^{-\alpha}}{\sigma^2}.
\end{align}

Since the source node can reliably connect to the satellite only when the SNR surpasses a minimum decoding threshold, i.e., $\text{SNR}(t)>\theta$, \eqref{equ:SNR_exprss} yields a maximum transmission distance for successful decoding, which can be expressed as
\begin{equation} \label{equ:r_max}
 r_\mathrm{max} =\min\left\{\left(\frac{P_{\mathrm{tx}}}{\sigma^2\theta}\right)^{\frac{1}{\alpha}},\sqrt{(R_E+h)^2-R_E^2}\right\}. 
\end{equation}
The second term in \eqref{equ:r_max} is the physical limit of \gls{LoS} obstruction caused by the Earth's curvature.
\color{black}
From the perspective of the source node, the maximum transmission distance separates the sky sphere into two parts, as depicted in Fig.~\ref{fig:sketch}: ($i$) the shaded area, corresponding to a dome in the sphere, whereby any satellite has a distance from the source node that is smaller than $r_\mathrm{max}$ and can thus establish a connection; ($ii$) the remaining area, containing satellites that are too far from the source node to provide connectivity. 

As illustrated in Fig.~\ref{fig:sketch}, the distance $r_\mathrm{max}$ is related to the node-centered zenith angle $\varphi_{ \mathrm{s} }$ and the Earth-centered zenith angle $ \varphi_{ \mathrm{e} }$. 
More concretely, for communication distance $r_\mathrm{max}$ strictly less than the maximum LoS distance, using the cosine theorem, $r_\mathrm{max}$ can be computed as
\begin{equation}
    r_\mathrm{max}^2=(R_E+h)^2+R_E^2-2R_E(R_E+h)\cos\varphi_e,
\end{equation}
where the earth-centered zenith angle $ \varphi_{ \mathrm{e}}$ is given by
\begin{equation}\label{eq:phie}
   \varphi_e= \cos^{-1}\left( \frac{(R_E+h)^2+R_E^2-\left( r_\mathrm{max}\right)^{2}}{2R_E(R_E+h)}\right),
\end{equation}
and the node-centered zenith angle $\varphi_{\mathrm{s}}$ is given by
\begin{align} \label{equ:phis}
    \varphi_{\mathrm{s}}&= \cos^{-1}\left( \frac{(R_E+h)\cos({ \varphi_{ \mathrm{e} })}-R_E}{r_\mathrm{max}}\right)\notag\\
    &= \cos^{-1}\left( \left( \frac{(2R_E+h)h}{2R_Er_\mathrm{max}}\right)-\frac{r_\mathrm{max}}{2R_E}\right),
\end{align}
where $r_\mathrm{max}$ is determined by \eqref{equ:r_max}. The above analyses (specifically, \eqref{eq:phie} and \eqref{equ:phis}) reveal that both the Earth-centered zenith angle $\varphi_e$ and the node-centered zenith angle $\varphi_s$ depend on the communication link parameters. Specifically, higher transmission power $P_\mathrm{tx}$ or a lower decoding threshold $\theta$ enables communication over longer distances, resulting in larger values of $\varphi_e$ and $\varphi_s$, thereby expanding the satellite's visibility region. Conversely, increasing the satellite altitude $h$ leads to greater path loss, reducing the maximum communication range at the same transmission power and consequently resulting in smaller effective angles $\varphi_e$ and $\varphi_s$.

From Fig.~\ref{fig:sketch}, we note that each satellite lying in the shaded dome region has a relative velocity with respect to the source node, owing to its orbital movement and the Earth's rotation. Such relative speed causes the satellite to pass through the shaded region rapidly. 
In what follows, we focus on the satellites within the dome region and their motion mode. While there are multiple satellite orbits, the maximum arc length within the specific dome region remains fixed. We assume that the velocity difference between the satellites and the source node remains constant. 
By considering relative motion, we can treat all satellites within the dome region as stationary in space, just as stars in the sky appear static for a short period. Meanwhile, the source node undergoes circular movements at a relative angular velocity $\omega$\cite{662636}:
\begin{align}\label{eq:omega}
    \omega &=\omega_s-\omega_e\cos(\, {i_o} \,)
    \nonumber\\
    &=\sqrt{\frac{GM_E}{(R_E+h)^3}}-\frac{2\pi}{T_E} \cos(\, {i_o} \,),
\end{align}
where $\omega_s$ denotes orbital angular velocity of satellite, $\omega_e$ represents the rotational angular velocity of earth, and $i_o$ is orbit plane inclination. $G$ denotes the gravitational constant, $M_E$ is the mass of the Earth, and $T_E$ represents the Earth’s rotational period.

If we view the on and off coverage stages of the source node from the perspective of a queuing system, then the duration for which a satellite remains within the dome region can be regarded as an on-service period. 
In contrast, the time between a satellite's departure from the dome region and the arrival of the subsequent satellite can be regarded as an off-service period. 
Due to the satellites' random positions, the durations of both the on- and off-service periods vary. 
Nonetheless, we can approximate these periods as independent and identically distributed. 
Moreover, we can leverage the techniques developed in \cite{madadi2017shared} to derive the distributions of the on- and off-service periods by accounting for the aforementioned approximations.

\begin{lemma} \label{lma:OnOff_dist}
The off-service periods follow an exponential distribution with density 
\begin{align}\label{eq:offperiod}
    \lambda_{ \mathrm{os}} = 2 \,\omega\,\lambda\,\sin (\varphi_{ \mathrm{e}}) \,\left(R_E+h\right)^2,
\end{align}
and the on-service periods have a common probability density function (PDF) given as
\begin{align}\label{eq:ONpdf}
f_{T_\mathrm{on}}(t)\!\!=\!\!
\begin{cases}\!
    \dfrac{\omega\cos(\varphi_{ \mathrm{e}})\tan(\frac{\omega t}{2})}{2 \varphi_{ \mathrm{e} }\sqrt{\sin^2 (\varphi_{ \mathrm{e} })-\sin^2(\frac{\omega t}{2}})}, & \mathrm{if}~t\!\in\! \left[0,\!\frac{2 \varphi_{ \mathrm{e} }}{\omega}\right],\! \\
    0, & \mathrm{otherwise},
\end{cases}
\end{align}
where $\varphi_{ \mathrm{e}}$ and $\omega$  are given in \eqref{eq:phie} and \eqref{eq:omega}, respectively.
\end{lemma}
\begin{IEEEproof}
Please see the Appendix \ref{proof:ONOFF}.
\end{IEEEproof}

An immediate observation from \eqref{eq:offperiod} is that the off-service periods shorten as the satellite density increases (because the average off-service period is $1/\lambda_{ \mathrm{os} }$). 
Similarly, these periods decrease as the earth-centered zenith angle of the visible spherical cap increases.

\subsection{Semi-Markov Modeling of the Energy Dynamics}
In this subsection, we construct a continuous-time semi-Markov process to model the state transitions of the energy buffer and satellite connectivity. 
Notably, due to the employed connectivity probing strategy (which consumes one unit of energy), the state transition of the energy buffer is affected by the connectivity state (e.g., whether a satellite is available or not upon the initiation of an update attempt).
Moreover, because the on period of satellite connections does not follow an exponential distribution (cf. \eqref{eq:ONpdf}), the memoryless property does not hold for the on state. 
To cope with this issue, in addition to the energy buffer and satellite connectivity states, we introduce a supplementary variable that captures the elapsed time since the beginning of the current on period of a satellite connection.
As such, we represent the system state at any time $t$ by the following tuple:
\begin{equation}
   X(t)=\big(S(t),E(t),A(t)\big),
\end{equation}
where $S(t) \in \{0,1\}$ denotes the satellite channel state, with $S(t)=0$ corresponding to the off state and $S(t)=1$ corresponding to the on state, $E(t)\in \{0,1,...B\}$ is the number of energy units stored in the buffer, and $A(t)\in[0,\infty)$ indicates the time elapsed since the on stage of the current connection. 
Note that $A(t)$ is only relevant when $S(t)=1$, while it is set to $0$ when $S(t)=0$.

Next, we characterize the state transitions over an infinitesimal time interval $dt$, under the assumption that at most one state transition occurs within $dt$. For ease of presentation, the system state is represented by $(1,e,a)$ when the satellite connection is on and by $(0,e)$ when it is off. The detailed transitions are given as follows:
\begin{itemize}
    \item If $S(t)=\text{1}$, the following state transitions are possible:
    \begin{itemize}
       \item[(1)] \textit{Energy harvesting}: If $e<B$, and an energy unit arrives with rate $\xi$, yielding the transition $(1,e,a)\to(1,e+1,a)$. 
           
        \item[(2)] \textit{Probe-then-transmission}: When $e \geq N+1$, the source node initiates a probe attempt and then transmission with rate $\mu$. Since the system is in the on state, the transmission succeeds and consumes $N+1$ energy units, resulting in $(1,e,a)\to(1,e-N-1,a)$.
           
       \item[(3)] \textit{On-to-off transition}: The on state terminates with hazard rate $h(a)$, leading to the transition $(1,e,a)\to(0,e)$. The hazard rate is given by
       \begin{equation}
       \begin{split}
          h(a)&\triangleq \mathrm{Pr}(T_\mathrm{on}\in[a,a+dt]~|~T_\mathrm{on}>a)\\
          &=\frac{f_{T_\mathrm{on}}(a)}{\bar{F}(a)},
       \end{split}
       \end{equation}
       in which
       \begin{equation}\label{eq:Fa}
    \bar{F}(a)=\frac{1}{\varphi_e}\mathrm{cos}^{-1}\left(\frac{\mathrm{cos}(\varphi_e)}{\mathrm{cos}(\omega a/2)}\right),~~0\leq a\leq T_\mathrm{max},
\end{equation}
where $T_\mathrm{max}=\frac{2 \varphi_{ \mathrm{e} }}{\omega}$. The quantity $h(a)$ represents the instantaneous transition rate of leaving the on state after having remained in the on state for $a$ seconds.
    \end{itemize}
    \item If $S(t)=\text{0}$, the following state transitions are possible:
    \begin{itemize}
        \item[(1)] \textit{Energy harvesting}: If $e<B$, and an energy unit arrives with rate $\xi$, yielding the transition $(0,e)\to(0,e+1)$.
        \item[(2)] \textit{Probe-then-silent}: If $e\geq N+1$, the source node initiates a probing attempt with rate $\mu$. Since the system is in the off state, the probe fails, and only one energy unit is consumed, resulting in $(0,e)\to (0,e-1)$
        \item[(3)] \textit{Off-to-on transition}: The off state transitions to the on state with rate $\lambda_{ \mathrm{os}}$, leading to $(0,e)\to(1,e,0)$, where the supplementary variable is reset to zero.
    \end{itemize}
\end{itemize}
Consequently, we can establish the energy transition rate matrix 
over an infinitesimal time interval in the satellite connection on stage, denoted by  
$\mathbf{Q}_1\in\mathbb{R}^{(B+1)\times(B+1)}$, as follows: 
\begin{equation}\label{eq:Q1}
   \begin{cases}
       [\mathbf{Q}_1]_{e,e+1}= \xi,~~~~~~e<B, \\
       [\mathbf{Q}_1]_{e,e-N-1}= \mu,~~e\geq N+1,\\
       [\mathbf{Q}_1]_{e,e}= -\xi\mathbb{I}(e<B)-\mu\mathbb{I}(e\geq N+1),
   \end{cases}
\end{equation}
where all the remaining entries are zero, and $\mathbb{I}(\mathcal{C})$ denotes the indicator function, which can be expressed as
\begin{equation}
    \mathbb{I}(\mathcal{C})=\begin{cases}
        1,~~&\text{if condition $\mathcal{C}$ holds,}\\
        0,~~&\text{otherwise}.
    \end{cases}
\end{equation}

Likewise, we construct the following energy transition rate matrix, $\mathbf{Q}_0\in\mathbb{R}^{(B+1)\times(B+1)}$, over an infinitesimal time interval in the satellite connection off stage
\begin{equation}\label{eq:Q0}
   \begin{cases}
       [\mathbf{Q}_0]_{e,e+1}= \xi,~~e<B, \\
       [\mathbf{Q}_0]_{e,e-1}= \mu,~~e\geq N+1,\\
       [\mathbf{Q}_0]_{e,e}= -\xi\mathbb{I}(e<B)-\mu\mathbb{I}(e\geq N+1)
   \end{cases}
\end{equation}
with the other entries set to zero.

Over a complete on (resp. off) period $t$ of satellite connection, as the connectivity state remains unchanged, the energy state dynamics according to $\mathbf{Q}_1$ (resp. $\mathbf{Q}_0$) during each infinitesimal time interval. 
Correspondingly, the energy state transition matrices over a duration $t$ of on and off states can be calculated as $\mathbf{P}_1(t)=e^{\mathbf{Q}_1 t}$ and $\mathbf{P}_0(t)=e^{\mathbf{Q}_0 t}$, respectively.

Aided by the above state transition rate matrices, we can analyze the steady-state behavior of the system. 
Specifically, the energy steady state probability of the off state is given by
\begin{equation}
    S_{e}^{0}=\lim_{t\to \infty}\mathrm{Pr}(S(t)=\text{0},E(t)=e),
\end{equation}
and energy steady state probability of on state is 
\begin{equation}
    S_{e}^{1}(a)da=\lim_{t\to \infty}\mathrm{Pr}(S(t)=\text{1}, E(t)=e, A(t)\in[a,a+da]).
\end{equation}

The supplementary-variable formulation provides an exact Markovian description by augmenting the joint channel-energy state with the elapsed time of on period. This turns the original finite-state model into a hybrid-state model with a continuous variable, and the stationary distribution is no longer obtained from finite-dimensional balance equations. It becomes a set of time-dependent density functions, whose evolution is governed by supplementary-variable differential equations with boundary integrals at the on-off switching epochs, making the direct solution highly complex. To facilitate the derivation, we consider the embedded process at the on-off switching epochs.

Particularly, using Lemma~\ref{lma:OnOff_dist}, we take an expectation on the on/off duration of the satellite connections, arriving at the following energy transition matrices across each (average) on and off connection period.
More precisely, from the beginning to the end of an on period, the energy transition matrix is
\begin{equation}\label{eq:P1}
    \mathbf{P}_1=\int_0^{T_{\max}} f_{T_\mathrm{on}}(t)e^{\mathbf{Q}_1t} dt,
\end{equation}
where $f_{T_\mathrm{on}}(t)$ is expressed in \eqref{eq:ONpdf}.
And the transition matrix of the energy state from the beginning to the end of an off period can be computed as
\begin{equation}\label{eq:P0}
\begin{split}
    \mathbf{P}_0 = \int_0^{\infty}\lambda_{ \mathrm{os}} e^{-\lambda_{ \mathrm{os}} t}e^{\mathbf{Q}_0t}dt=\lambda_{ \mathrm{os}}(\lambda_{ \mathrm{os}}\mathbf{I}-\mathbf{Q}_0)^{-1}.
    \end{split}
\end{equation}

Consequently, the steady state of energy distribution, denoted by $\boldsymbol{\alpha}$, at the beginning of a typical off period can be obtained by solving the following system of equations
\begin{subequations}
\begin{align}
&\boldsymbol{\alpha}=\mathbf{\boldsymbol{\alpha}}\mathbf{P}_0 \mathbf{P}_1, \label{EMC:trans}\\
&\boldsymbol{\alpha1}=1\label{EMC:sum1},
\end{align}
\end{subequations}
where \eqref{EMC:trans} is the state transition of an on-off process cycle, and \eqref{EMC:sum1} is the normalization condition.

We then consider the occupation time of each energy state in one off period. After time $t$, the system is still in the off period, and the energy-state probability is $\boldsymbol{\alpha} e^{-\lambda_{ \mathrm{os}} t}e^{\mathbf{Q}_0t}$. Over a small time interval $[t,t+dt]$, the occupation time of each energy state is $\boldsymbol{\alpha} e^{-\lambda_{ \mathrm{os}} t}e^{\mathbf{Q}_0t}dt$. By accumulating this term over all $t\geq0$, the expected occupation-time vector of each energy state $e$ during one off period can be given by
\begin{equation}
   \mathbf{T}^{0}= \boldsymbol{\alpha} \int_0^{\infty} e^{-\lambda_{ \mathrm{os}} t}e^{\mathbf{Q}_0t}dt=\boldsymbol{\alpha}(\lambda_{ \mathrm{os}} \mathbf{I}-\mathbf{Q}_0)^{-1},
\end{equation}
and the sum of the occupation-time components $\mathbf{T}^{0}\mathbf{1}=1/\lambda_\mathrm{os}$.

Furthermore, the expectation of one on-off duration can be expressed as
\begin{equation}
    \mathbb{E}[L]=1/\lambda_{ \mathrm{os}}+\mathbb{E}[T_\mathrm{on}].
\end{equation}
Therefore, the steady state probability $\mathbf{S}^{0}=[S^0_0,S^0_1,...,S^0_B]$ can be given by
\begin{equation}\label{eq:S0}
    \mathbf{S}^{0}=\frac{\mathbf{T}^{0}}{\mathbb{E}[L]} =\frac{\boldsymbol{\alpha}(\lambda_{ \mathrm{os}} \mathbf{I}-\mathbf{Q}_0)^{-1}}{1/\lambda_{\mathrm{os}}+\mathbb{E}[T_\mathrm{on}]}.
\end{equation}
By using the fact $\mathbf{S}^{1}(a)=\lambda_{ \mathrm{os}}e^{\mathbf{Q}_1 a}\bar{F}(a)\mathbf{S}^{0}$, the steady state probability $\mathbf{S}^{1}(a)=[S^1_0(a),S^1_1(a),...,S^1_B(a)]$ can be given by
\begin{equation}\label{eq:S1}
     \mathbf{S}^{1}(a)=\frac{\boldsymbol{\alpha} \lambda_{ \mathrm{os}} e^{\mathbf{Q}_1 a}\bar{F}(a)(\lambda_{ \mathrm{os}} \mathbf{I}-\mathbf{Q}_0)^{-1}}{1/\lambda_{\mathrm{os}}+\mathbb{E}[T_\mathrm{on}]}.
\end{equation}

The basic idea behind the derivation from \eqref{eq:P1} to \eqref{eq:S1} is to first characterize the steady-state distribution of the embedded chain observed at the beginnings of off periods, and then use the off holding-time distributions to recover the steady-state distribution at an arbitrary time. By working with the embedded chain at renewal epochs, we avoid directly tracking the non-Markovian evolution within the on state induced by the general holding-time distribution.

\subsection{Analysis of the Time-Average AoI}
With the preparation above, we derive the time-average AoI in this part. 
Let $u_e$ denote the mean waiting time to the next successful update
starting from the off state with energy level $e$, and let $v_e(a)$
denote the corresponding mean waiting time starting from the on state
with energy level $e$ and on holding time $a$. Then, based on the steady-state results obtained from the semi-Markov analysis, the time-average AoI is given by
\begin{equation}\label{eq:AoI:EHMCanalysis}
    \bar{\Delta} = \sum_{e=0}^{B}S_e^0u_e+
\int_0^{T_{\max}}\sum_{e=0}^{B}S_e^1(a)v_e(a)\,da+3D.
\end{equation}

By applying the first-passage-time method, the analytical expression of the time-average AoI can be obtained as follows.
\begin{theorem}\label{Th:AoI:semi}
  The time-average AoI can be computed as
 \begin{equation}
\bar{\Delta}=\mathbf{S}^0\mathbf{u}+\lambda_{\mathrm{os}} \mathbf{S}^0
\int_0^{\frac{2\varphi_{ \mathrm{e} }}{\omega}}
e^{\mathbf{Q}_1a}\mathbf w_1(a)\,da+3D,
\end{equation}
where the off period waiting time vector $\mathbf{u}=[u_0,u_1,...,u_B]^{T}$ is given by
\begin{equation}\label{eq:u}
\mathbf{u}=-(\mathbf{Q}_0-\lambda_{\mathrm{os}} \mathbf{I}+\lambda_{\mathrm{os}} \mathbf{A})^{-1}
(\mathbf{1}+\lambda_{\mathrm{os}}\mathbf b),
\end{equation}
and
$$\mathbf{A}\!=\!\int_0^{\frac{2\varphi_{ \mathrm{e} }}{\omega}}e^{\mathbf{C}t}f(t)\,dt,
~~~
\mathbf{b}\!=\!\int_0^{\frac{2\varphi_{ \mathrm{e} }}{\omega}}e^{\mathbf{C}t}\bar F(t)\mathbf 1\,dt,
$$
with $\mathbf{w}_1(a)$ given by
\begin{equation}\label{eq:w1}
\mathbf{w}_1(a)=
\int_{a}^{\frac{2\varphi_{ \mathrm{e} }}{\omega}}e^{\mathbf{C}(t-a)}\left(\bar F(t)\mathbf 1+f(t)\mathbf{u}\right)dt,
\end{equation}
in which the entries in matrix $\mathbf{C}$ are
\begin{equation}\label{eq:C:define}
  [\mathbf{C}]_{e,j}\!=\!\begin{cases}
\xi, \quad\quad\quad j=e+1,~0\le e<B,\\[0.6ex]
-\xi\mathbb{I}(e\!<\!B)\!-\!\mu\mathbb{I}(e\!\geq\!N\!+\!1),
\quad j=e,~0\le e\le B,\\[0.6ex]
0,
\quad\quad\quad \text{otherwise},
\end{cases}
\end{equation}
and $\mathbf{Q}_0$, $\mathbf{Q}_1$, and $\mathbf{S}^{0}$ are given by \eqref{eq:Q0}, \eqref{eq:Q1} and \eqref{eq:S0}, respectively.
\end{theorem}
\begin{IEEEproof}
      Please see the Appendix \ref{eq:AoI:matrix}.
\end{IEEEproof}
Based on Theorem \ref{Th:AoI:semi}, we further present two special cases, i.e., the exponential distribution assumption and the energy-sufficient regime.

\begin{corollary}
 When the energy is sufficient, vector $\mathbf{S}^{0}$ degenerates into the probability $S^{0}$ that the system is in the off state at an arbitrary time, which is given by
 \begin{equation}
  S^{0} =1/(1+\lambda_\mathrm{os}\mathbb{E}[T_\mathrm{on}]),  
 \end{equation}
 and \eqref{eq:AoI:EHMCanalysis} can be simplified as
\begin{align}\label{eq:AoI:EnergyAlways}
\bar{\Delta} 
\!=\!\frac{1}{\mu}\!+\!\frac{1}{1\!+\!\lambda_{\mathrm{os}}\mathbb{E}[T_\mathrm{on}]}\left(\frac{1}{\mu}\!+\!\frac{1}{\lambda_{\mathrm{os}}(1-\mathcal{L}_\mathrm{on}(\mu))}\right)\!+\!3D,
\end{align}
where $\mathcal{L}_{\mathrm{on}}(\mu)=  \int_{0}^{\frac{2\varphi_{ \mathrm{e} }}{\omega}} e^{-\mu t} f_{T_\mathrm{on}}(t)\,dt.$
\end{corollary}
\begin{IEEEproof}
 In this regime, the energy dimension vanishes and all energy-state transition rate matrices collapse to scalars, i.e., $\mathbf{Q}_0=0$, $\mathbf{Q}_1=0$, and $\mathbf{C}=-\mu$. The off period waiting time vector $
\mathbf{u}=\frac{1}{\lambda_{\mathrm{os}}\left(1-\mathcal L_{\mathrm{on}}(\mu)\right)}
+\frac{1}{\mu}$, and $\lambda_{\mathrm{os}} \mathbf{S}^0
\int_0^{\frac{2\varphi_{ \mathrm{e} }}{\omega}}
e^{\mathbf{Q}_1a}\mathbf w_1(a)\,da=\lambda_{\mathrm{os}} S^{0}\int_0^{\frac{2\varphi_{ \mathrm{e} }}{\omega}}w_1(a)da=\frac{1}{\mu}$. Applying these identities to the transition matrices and stationary distributions yields the stated result.
\end{IEEEproof}

The result is consistent with \cite{yanwuGlobecom}, except for the propagation delay, since the probe overhead also needs to be considered.

\begin{corollary}
   When the on period of satellite connections is approximated by an exponential distribution with the same average holding time, \eqref{eq:AoI:EHMCanalysis} can be simplified as
  \begin{align}
    \bar{\Delta}=&\lambda_{\mathrm{os}}\mathbf{S}^0(\lambda_1 \mathbf{I}-\mathbf{Q}_1)^{-1}(\lambda_1\mathbf{I}-\mathbf{C})^{-1}(\mathbf{1}+\lambda_1 \mathbf{u})+\mathbf{S}^0\mathbf{u}\!+\!3D,
  \end{align} 
  where $\mathbf{u}$ is given by
  \begin{equation}
    \mathbf{u}=-((\lambda_1\mathbf{I}-\mathbf{C})(\mathbf{Q}_1-\lambda_\mathrm{os}\mathbf{I})+\lambda_\mathrm{os}\lambda_1 \mathbf{I})^{-1}((\lambda_1\mathbf{I}-\mathbf{C})\mathbf{1}+\lambda_\mathrm{os}\mathbf{1}),
  \end{equation}
and $\lambda_1$ is the rate of the exponential distribution with the same mean as the original on period, i.e., $\lambda_1=1/\mathbb{E}[T_\mathrm{on}]$. 
\end{corollary}

\begin{IEEEproof}
Under the exponential approximation, the satellite connection on period becomes memoryless. Hence, the on holding time variable is no longer needed, and the semi-Markov process of the connection and energy states reduces to a finite-state CTMC. Substituting $\bar F(t)=e^{-\lambda_1t}$ and $f(t)=\lambda_1e^{-\lambda_1t}$ into the semi-Markov expressions converts all on holding time integrals into matrix resolvents. Specifically, 
$\int_0^{\infty} e^{\mathbf A t}f(t)dt=
\lambda_1(\lambda_1\mathbf I-\mathbf A)^{-1},$ and $\int_0^{\infty} e^{\mathbf A t}\bar F(t)dt=(\lambda_1\mathbf I-\mathbf A)^{-1}$. 
Applying these identities to the transition matrices, stationary distributions, and first-passage-time equations gives the stated result.
\end{IEEEproof}
The exponential assumption is a standard modeling simplification that reduces analytical complexity. By ignoring the determinisn of the satellite trajectory induced on process, this approximation converts the integral terms associated with the on period into matrix resolvents.

\subsection{Low-Complexity Time-average AoI Approximation}
It can be observed that jointly characterizing the channel and energy dynamics yields an accurate model but requires computing several matrix integrals. As an engineering approximation, we next consider using the mean on holding probability $P_\mathrm{on}$ and adopt a mean-field approach to derive a low-complexity approximation, which will be compared with the previously developed exact model.

Beyond providing a simple approximation, we aim to investigate whether, when a high-dimensional CTMC describing two coupled processes is analytically intractable, the detailed state transitions of one process can be replaced by its steady-state probability to reduce complexity, and under what conditions such a reduction is valid.

Specifically, the steady state probability of the source node being in the on service state can be computed as
\begin{equation}\label{eq:Pon}
 \mathrm{P_{on}}\!=\!\frac{ 2 \,\omega\,\lambda\,\sin (\varphi_{ \mathrm{e}}) \,\left(R_E+h\right)^2\int_{0}^{\frac{2 \varphi_{ \mathrm{e} }}{\omega}} tf_{T_\mathrm{on}}(t)dt}{1+ 2 \,\omega\,\lambda\,\sin (\varphi_{ \mathrm{e}}) \,\left(R_E+h\right)^2\int_{0}^{\frac{2 \varphi_{ \mathrm{e} }}{\omega}}tf_{T_\mathrm{on}}(t)dt}.
\end{equation}

Based on the result of the satellite on service periods, we approximate the successful probing probability by the long-term on probability of the satellite channel. This allows us to further characterize the energy dynamics of the node while decoupling it from the on–off state transition process.

\color{black}
\begin{figure}[t]
    \centering
    \includegraphics[width=\linewidth]{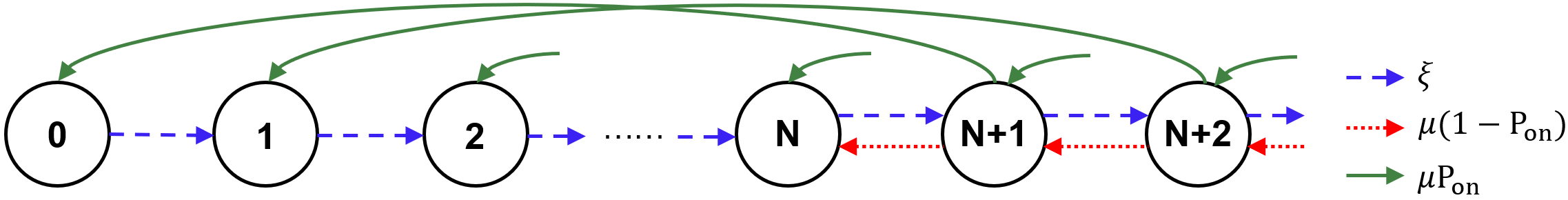}
    \caption{State-transition diagram of the approximation Markov chain describing the dynamics of energy harvesting and consumption.}
    \label{fig:MarkovChain}
    \vspace{-0.5cm}
\end{figure}
We then establish a one-dimensional continuous-time Markov chain (CTMC) to model the state transitions of the energy buffer, based on which we derive an analytical expression for the steady-state distribution.

Specifically, the state space of the CTMC is represented by the energy level $i \in \mathbb{N}$, where $i$ corresponds to the number of energy units stored in the buffer. The state transitions of the CTMC are summarized as follows:
\begin{itemize}
    \item If $0\leq i\leq N$: the energy storage is insufficient to support a status update (which includes connection probing and data transmission). The only possible state transition is $i\to i+1$ with rate $\xi$, representing an energy unit arrival.
    \item If $i \geq N+1 $: the energy storage is sufficient to support (at least) a status update, and an incoming energy unit can be stored in the energy buffer. Three transitions are possible:
    \begin{itemize}
        \item[1)] $i\to i+1$ with rate $\xi$: This corresponds to the arrival of a new energy unit.
        \item[2)] $i\to i-1$ with rate $\mu(1-\mathrm{P_{on}})$: This occurs when the node consumes one energy unit for probe but finds no available satellite in the visibility region (an off state).
        \item[3)] $i\to i-N-1$ with rate $\mu \mathrm{P_{on}}$: This occurs when the node consumes one energy unit for probe, finds an available satellite (an on state), and subsequently expends $N$ additional units for transmission.
    \end{itemize}

     \item If $ i= B$: the energy storage is sufficient to support (at least) a status update, but the incoming energy unit cannot store at energy buffer due to the energy buffer is full. In this scenario, two transitions are possible:
          \begin{itemize}
              \item[1)] $i\to i-1$ with rate $\mu(1-\mathrm{P_{on}})$: This occurs when the node consumes one energy unit for probe but finds no available satellite in the visibility region (an off state).
         \item[2)] $i\to i-N-1$ with rate $\mu \mathrm{P_{on}}$: This occurs when the node consumes one energy unit for probe, finds an available satellite (an on state), and subsequently expends $N$ additional units for transmission.
          \end{itemize}
\end{itemize}
Consequently, the dynamics of energy harvesting and consumption are captured by the graphical representation in Fig.~\ref{fig:MarkovChain}. Since the CTMC has a finite state space and is irreducible, all states are positive recurrent, and there exists a unique stationary distribution.
As such, we denote by $S_i$ the probability that the state of the energy buffer of node is $i$ when the system enters the steady state. The following lemma provides an analytical characterization of this quantity.

\begin{lemma}
  The steady-state distribution of the energy buffer state with finite energy buffer capacity can be approximated by \eqref{eq:S0:app}
  \begin{figure*}
     \begin{equation}\label{eq:S0:app}
      S_0=\left[\frac{N+1}{1-z}-\frac{z(1-z^{N+1})}{(1-z)^2}+\frac{1-P_{\mathrm{on}}}{P_{\mathrm{on}}}+\frac{\xi}{\mu P_{\mathrm{on}}}\frac{1-z^{B-2N-1}}{1-z}+\Gamma_B
\frac{\left(\frac{\mu}{\xi}-r_2\right)\frac{1-r_1^{N+1}}{1-r_1}+\left(r_1-\frac{\mu}{\xi}\right)
\frac{1-r_2^{N+1}}{1-r_2}}{r_1-r_2}
\right]^{-1}
  \end{equation} 
  \hrule
  \vspace{-0.4cm}
  \end{figure*}
and
\begin{equation}
S_i\!=\!
\begin{cases}
\dfrac{1-z^{i+1}}{1-z}S_0,
& 0\le i\le N-1,\\[2.5ex]
\left(
\dfrac{1-z^{N+1}}{1-z}
+\dfrac{1-P_{\mathrm{on}}}{P_{\mathrm{on}}}
\right)S_0,
& i=N,\\[3ex]
\dfrac{\xi}{\mu P_{\mathrm{on}}}S_0z^{i-N-1},
& N\!+\!1\le i\le B\!-\!N\!-\!1,\\[3ex]
\Gamma_B\Psi_{B-i}S_0,
& B-N\le i\le B,
\end{cases}
\end{equation}
where $\mathrm{P_{on}}$ is given in \eqref{eq:Pon} and $z$ is the non-negative dominant root of the following equation:
  \begin{equation} \label{equ:root_z}
    \mu \mathrm{P_{on}}z^{N+2}+\mu (1-\mathrm{P_{on}})z^2-(\xi+\mu)z+\xi =0,
\end{equation}
and $\Gamma_B$, $\Psi_i$, $r_{1}$, and $r_{2}$ are given respectively as follows:
\begin{equation}\label{eq:GammaB}
    \Gamma_B=\frac{\xi}{\mu P_{\mathrm{on}}}\cdot\frac{(1-P_{\mathrm{on}})z^{B-2N-1}+P_{\mathrm{on}}z^{B-N-1}}{(1-P_{\mathrm{on}})\Psi_N+P_{\mathrm{on}}},
\end{equation}

\begin{equation}
    \Psi_i=\frac{
\left(\frac{\mu}{\xi}-r_2\right)r_1^{i}+\left(r_1-\frac{\mu}{\xi}\right)r_2^{i}}{r_1-r_2},
\end{equation}
\begin{equation}\label{eq:r1}
    r_{1}=\frac{\xi+\mu +\sqrt{(\xi+\mu)^2-4\xi\mu(1-P_{\mathrm{on}})}}{2\xi},
\end{equation}
\begin{equation}\label{eq:r2}
    r_{2}=\frac{\xi+\mu -\sqrt{(\xi+\mu)^2-4\xi\mu(1-P_{\mathrm{on}})}}{2\xi}.
\end{equation}
\end{lemma}
\begin{IEEEproof}
    Please see the Appendix \ref{eq:CTMC}.
\end{IEEEproof}
Therefore, the probability that the node has sufficient energy to transmit can be given by
\begin{equation}\label{eq:PE:def}
   \mathrm{P_E}\!
\approx\!
\left(\frac{\xi}{\mu P_{\mathrm{on}}}\frac{1-z^{B-2N-1}}{1-z}
+
\Gamma_B\sum_{m=0}^{N}\Psi_m\right)
S_0.
\end{equation}
\color{black}
When the energy buffer size $B$ goes to infinity, and the steady condition of the CTMC $z<1$ can be satisfied, the result can be simplified as the following corollary.

\begin{corollary}
With infinite buffer capacity, the steady-state distribution of the energy buffer state is 
\begin{equation}
   S_i\!=\!\begin{cases}
      \dfrac{\mathrm{P_{on}}(1-z^{i+1})}{1+N \mathrm{P_{on}}},~i \leq N-1,\\ \\

    \!\!\left(\dfrac{1\!-\!z^{N+1}}{1\!-\!z}\!+\!\dfrac{1-P_{\mathrm{on}}}{P_{\mathrm{on}}}\right)\!\dfrac{\mathrm{P_{on}}\xi(1\!-\!z)}{\mu(1+NP_{\mathrm{on}})},~i\!=\!N\\ \\
      
       \dfrac{\xi (1-z) }{\mu (1+N \mathrm{P_{on}})}z^{i-N-1},~i\geq N+1,
   \end{cases} 
\end{equation}
where $\mathrm{P_{on}}$ is given in \eqref{eq:Pon} and $z\in(0,1)$ is the root of the following equation
\begin{equation} \label{equ:root_z}
    \mu \mathrm{P_{on}}z^{N+2}+\mu (1-\mathrm{P_{on}})z^2-(\xi+\mu)z+\xi =0.
\end{equation}
\end{corollary}

Notably, the necessary condition of finding a root of \eqref{equ:root_z} in $(0, 1)$ is $\mu(1+N \mathrm{P_{on}})>\xi$, i.e., when the energy arrival rate is lower than the energy consumption rate on average. In this case, the node operates in an energy-constrained regime. Otherwise, the source node always has sufficient energy to transmit.

Therefore, the probability that the node has sufficient energy to transmit can be given by
\begin{align}\label{eq:Pe}
\mathrm{P_E}= \min\left\{\frac{\xi}{\mu(1+N P_{\mathrm{on}})},1\right\}.
\end{align}

On the other hand, if the source node performs the status update without probing the coverage condition, the probability that the node has sufficient energy to transmit a data packet (denoted by $\mathrm{\bar{P}_E}$, which is the probability that the node has accumulated at least $N$ units of energy from the environment) can be computed by
    \begin{equation}\label{eq:Pe:nodetect}
     \mathrm{\bar{P}_E}=\min\left\{\frac{\xi}{N\mu},1\right\}.   
    \end{equation}

By comparing \eqref{eq:Pe} and \eqref{eq:Pe:nodetect}, we observe that when the satellite's on-off state is unknown to the node, a strategy involving coverage probing (which consumes one energy unit) before potential transmission enhances energy efficiency under the condition $\mathrm{P_{on}}\leq 1-\frac{1}{N}$. This condition implies scenarios where the energy consumption of the payload is relatively high compared to probing, and the channel availability (on state probability) is relatively low.

Leveraging the on-off process of the satellite link and the energy level results of the node from the Markov steady-state analysis, we can proceed to analyze the \gls{AoI} performance. 

The energy harvesting constraint can be captured by thinning the update attempt process with the energy-sufficient probability. Hence, the refined attempt rate can be regarded as $\mu_{\mathrm{eff}}=\mu P_E$, where $P_E$ is the energy sufficient probability. Using the AoI evaluation without energy constraint, i.e., \eqref{eq:AoI:EnergyAlways}, and effective service rate $\mu P_E$, we can obtain the follows.
\begin{theorem}\label{Aporox:AAoI}
The time-average AoI in probe-then-transmission scheme can be approximately given by
\begin{equation}
\begin{split}
  \bar{\Delta}&\approx\frac{1}{1+\lambda_{ \mathrm{os}}\mathbb{E}[T_{\mathrm{on}}]}\!\left(\frac{1}{\mu P_E}+\frac{1}{\lambda_{\mathrm{os}}(1-\mathcal{L}_{\mathrm{on}}(\mu P_E))}\right)\\&~~~~~+\frac{1}{\mu P_E}+3D,    
\end{split}
\end{equation}
where $\mathbb{E}[T_{\mathrm{on}}]=\int_{0}^{\frac{2\varphi_{\mathrm{e}}}{\omega}} t f_{T_{\mathrm{on}}}(t) dt$, $\varphi_{ \mathrm{e}}$, $\omega$, $\lambda_{\mathrm{os}}$, $f_{T_{\mathrm{on}}}(t)$ and $P_E$ are given by \eqref{eq:phie}, \eqref{eq:omega}, \eqref{eq:offperiod}, \eqref{eq:ONpdf} and \eqref{eq:PE:def}, respectively, and 
\begin{equation}
  \mathcal{L}_{\mathrm{on}}(\mu P_E)=  \int_{0}^{\frac{2\varphi_{ \mathrm{e} }}{\omega}} e^{-\mu P_E t}f_{T_{\mathrm{on}}}\,dt.
\end{equation}   
\end{theorem}
\begin{remark}
Theorem 2 approximates the impact of energy harvesting by thinning the effective update rate. While this captures the rate loss caused by energy scarcity, pure thinning tends to overestimate the variability of the inter-update time, since the energy accumulation process partially regularizes update attempts. We use a subtractive correction term $\delta\!=\!N/2\xi$ to compensate for this overestimation, as validated in Fig.~\ref{fig:AoIvsEHRate_combined} and Fig.~\ref{fig:11}, which is the residual waiting time gap between the memoryless Bernoulli-thinning approximation and the Erlang energy-accumulation process\cite{JSAC2025}.
\end{remark}
Based on theorem \ref{Aporox:AAoI}, we further present some special cases:
\subsubsection{Infinite Energy Buffer $(B \rightarrow \infty)$} The time-average AoI can be further simplified by
\begin{equation}\label{eq:AoI:Probe:ECR}
\begin{split}
   \bar{\Delta}|_{B\to\infty}\!\approx\!\frac{1}{1\!+\!\lambda_{\mathrm{os}}\mathbb{E}[T_{\mathrm{on}}]}\!\left(\psi_{P}\!+\!\frac{1}{\lambda_{\mathrm{os}}(1\!-\!\mathcal{L}_{\mathrm{on}}\left(\psi_{P}\right))}\right)\!+\!\psi_{P}\!+\!3D,
\end{split}
\end{equation}
where $\psi_{P}=\max\{\frac{1+N P_{\mathrm{on}}}{\xi}, \frac{1}{\mu}\}$.

\subsubsection{Small Service Rate $(\mu \to 0)$} The system always operates in the energy-sufficient regime, $\psi_{P}=1/\mu$. By assuming that the propagation delay is negligible compared with an extremely large update interval, we obtain
\begin{equation}
  \bar{\Delta}|_{\mu\to 0}\approx\frac{1}{\mu \mathrm{P_{on}}},
\end{equation}
where $\mathrm{P_{on}}$ is given in \eqref{eq:Pon}. This implies that, in the case of relatively sparse updates, the time-average AoI is primarily determined by the update rate, on-state probability (coverage probability). 

\subsubsection{Large Satellite Density $(\lambda \to \infty)$} 
We have $\mathrm{P_{on}}\to 1$, we can obtain
\begin{equation}\label{eq:alwaysON}
\begin{split}
   \bar{\Delta}|_{\lambda \to \infty}\approx\max\left\{\frac{N+1}{\xi},~\frac{1}{\mu}\right\}+3D.
\end{split}
\end{equation}
This indicates that, when satellites are densely deployed and coverage is good, a behavior similar to terrestrial networks emerges: as the update rate increases, the AoI initially decreases, dominated by the update interval, and then levels off, becoming dominated by energy constraints\cite{JSAC2025}. Therefore, in this scenario, it is sufficient to ensure that the energy consumption rate exceeds the energy arrival rate, i.e., $\mu\geq \frac{\xi}{1+N}$.

Based on the analytical framework and $\mathrm{\bar{P}_E}$, we can also derive the time-average AoI under the direct transmission scheme without probing, which uses $N$ energy units for each data packet transmission, given by
\begin{equation}\label{eq:AoI:BlindTrans}
\begin{split}
   \bar{\Delta}_{D}\!=\!\frac{1}{1+\lambda_{ \mathrm{os}}\mathbb{E}[T_{\mathrm{on}}]}\!\left(\psi_{D}+\frac{1}{\lambda_{\mathrm{os}}(1-\mathcal{L}_{\mathrm{on}}\left(\psi_{D}\right))}\right)\!+\!\psi_{D}\!+\!D,
\end{split}
\end{equation}
where $\psi_{D}=\max\{\frac{N}{\xi}, \frac{1}{\mu}\}$.

Comparing \eqref{eq:AoI:Probe:ECR}, \eqref{eq:alwaysON}, and \eqref{eq:AoI:BlindTrans}, we can readily see that when the $P_{\mathrm{on}}$ is small and $\frac{N}{\xi}$ is large, the probe scheme is preferred.  More specifically, the probe mechanism can reduce the waiting time by avoiding ineffective energy consumption, but it introduces an additional delay cost of $2D$. Therefore, the probe scheme outperforms blind transmission only when the waiting-time reduction, jointly determined by the energy-saving gain and the on-period duration distribution, exceeds the additional $2D$.

\color{black}
\section{Numerical Results and Discussions}
In this section, we evaluate the time-average AoI across different network configurations using the analytical results. 
Unless otherwise specified, we use the following parameters: $B=3N+1$, $R_E=6371~\mathrm{km}$, $P_\mathrm{tx}=30~\mathrm{dBm}$, $\sigma^2=-105~\mathrm{dBm}$, $\alpha=2$, $i_o=53^{\circ}$, $h=800~\mathrm{km}$, $GM_E=3.986\times10^{14}~\mathrm{m^3/s^2}$, and $T_E=86400~\mathrm{s}$.

Fig.~\ref{fig:AoIvsheight} investigates the impact of satellite altitude on the time-average AoI, considering different total numbers of satellites. A clear trend is that the time-average AoI deteriorates as the satellite orbital altitude increases. This finding underscores the inherent advantage of the LEO satellite for applications demanding high information freshness, suggesting that minimizing orbital altitude should be prioritized where feasible. The reason is that, under the same transmit power, the maximum propagation range of a ground node is fixed; as the orbital altitude increases, the satellite’s elevation angle within the node’s visibility region decreases, thereby reducing the probability that the link remains in the on state. In addition, higher orbital altitudes inherently lead to larger propagation delays, which further degrade AoI. On the other hand, increasing satellite density can improve timeliness and mitigate the severe degradation caused by orbital altitude differences.
\begin{figure}[t!] 
  \centering{}
    {\includegraphics[width=0.9\figwidth]{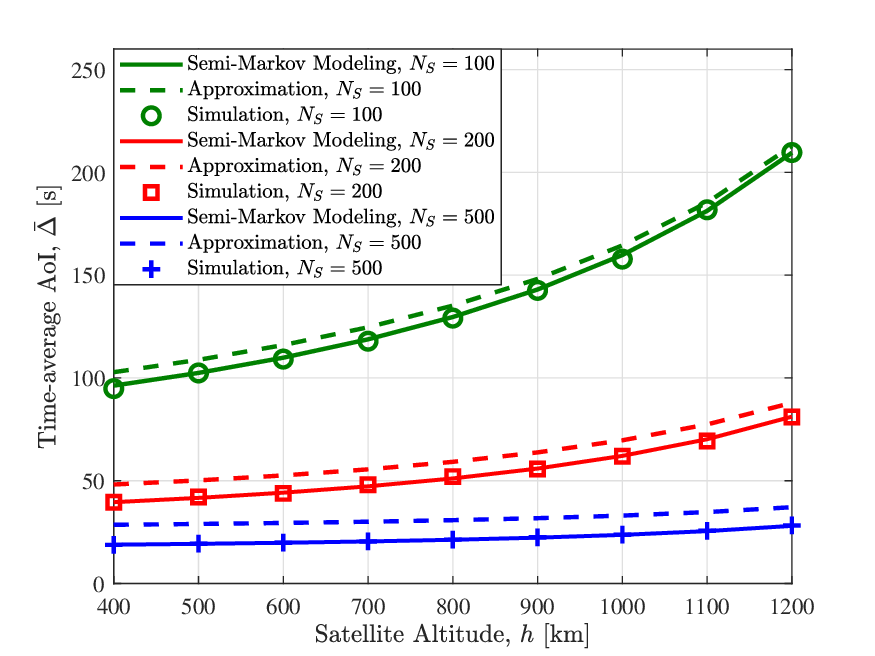}}
  \caption{Time-average AoI versus the satellite altitude. $\mu=0.2~ \mathrm{packet/s}$, $\xi=0.5~\mathrm{energy~units /s} $, $N=10~\mathrm{energy~units/packet}$, $\theta=10\mathrm{dB}$.}
  \label{fig:AoIvsheight}
   \vspace{-0.5cm}
\end{figure}

Fig.~\ref{fig:6} shows the time-average AoI as a function of the number of satellites,  for different values of the SNR decoding thresholds and update rates. As the number of satellites increases, the time-average AoI decreases and gradually saturates. This is because denser satellite deployment shortens the mean off-service period and increases the fraction of time the source is in the on-service state. Consequently, update attempts are more likely to occur during satellite-available periods. However, the marginal AoI reduction becomes smaller as the satellite density further increases, since the AoI eventually becomes limited by the update rate and the energy-availability constraint rather than by satellite intermittency alone, and is close to the value of $1/(\mu \mathrm{P_E})$.

\begin{figure}[t] 
  \centering{}
   {\includegraphics[width=0.9\figwidth]{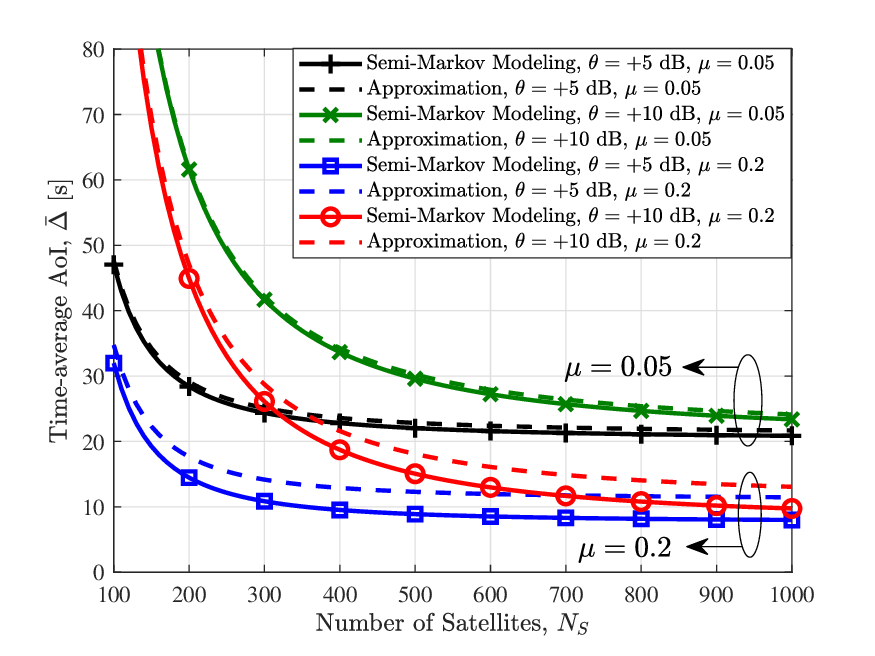}}
  \caption{Time-average AoI versus number of satellites. $\xi=1~\mathrm{energy~units/s}$, $N=10~\mathrm{energy~units/packet}$.}
  \label{fig:6}
  \vspace{-0.5cm}
\end{figure}

\begin{figure}[t] 
  \centering{}
    {\includegraphics[width=0.9\figwidth]{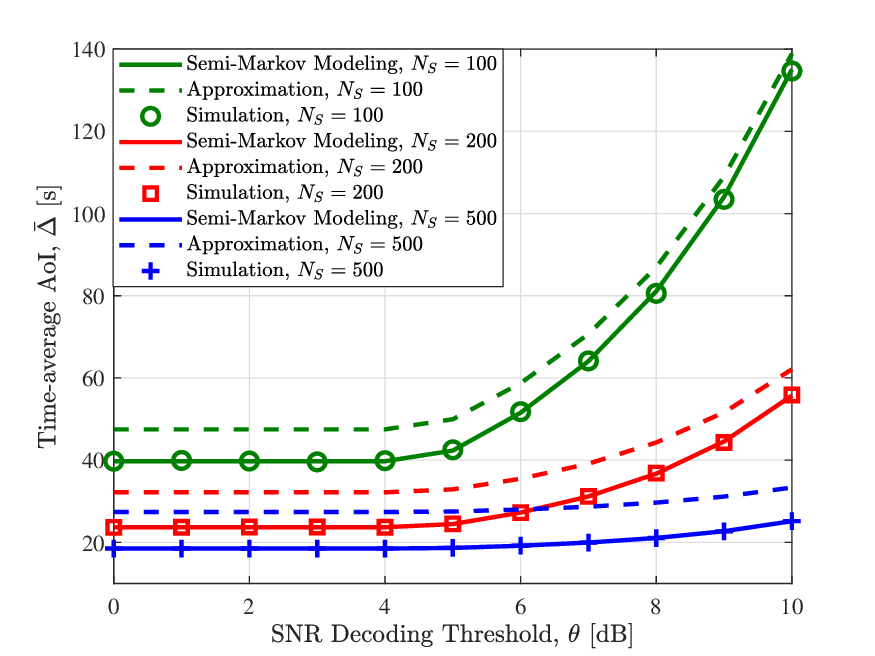}}
  \caption{Time-average AoI versus SNR decoding threshold.  $\mu=0.1~\mathrm{packet/s}$, $\xi=0.5~\mathrm{energy~units /s}$, $N=10~\mathrm{energy~units/packet}$.}
  \label{fig:7}
  \vspace{-0.5cm}
\end{figure}

To closely examine the effect of the decoding threshold, Fig.~\ref{fig:7} plots the time-average AoI as a function of the minimum required SNR for successful decoding. The performance curve can be divided into two distinct regions: 

$\bullet$ $\theta\in[0\mathrm{dB},5\mathrm{dB}]$: In this range, the decoding threshold is already sufficiently low. Further decreasing the threshold (which in turn increases the maximum transmission distance) yields no additional performance gain. Therefore, in this regime, the AoI performance is insensitive to the decoding threshold and is instead primarily dominated by the update frequency and satellite density. 

$\bullet$ $\theta \!\in\![5\mathrm{dB},10\mathrm{dB}]$: Here, as the SNR decoding threshold increases, the maximum transmission distance effectively decreases. This shrinks the visible range of the node (i.e., the spherical cap), leading to more frequent and severe intermittent link outages. This limitation becomes the bottleneck, directly causing the AoI to degrade (increase) as the SNR decoding threshold rises. Furthermore, increasing the number of satellites can alleviate this problem. As shown in the figure, for the curve with a larger number of satellites ($N_S = 500$), the rate of AoI degradation is noticeably lower, as the higher satellite density partially compensates for the reduced visibility.


Next, we examine the impact of the energy arrival rate, the energy consumption per payload transmission, and the update rate on the time-average AoI, as these parameters directly govern whether the system operates in an energy-sufficient or an energy-constrained regime.
\par Fig.~\ref{fig:AoIvsEHRate_combined} depicts the time-average AoI as a function of the energy arrival rate at the ground node, illustrating how the availability of energy for transmission impacts AoI. For small update rates in Fig.~\ref{fig:AoIvsEHRate_combined}, i.e., $\mu=0.05$, the AoI initially decreases when $\xi\leq 0.5$, and then gradually levels off as the energy arrival rate increases. This indicates that the impact of the energy arrival rate on the AoI is regime-dependent: it is significant in the energy-constrained regime but becomes marginal once the system enters the energy-sufficient regime.
\begin{figure}[t] 
  \centering{}
    {\includegraphics[width=0.9\figwidth]{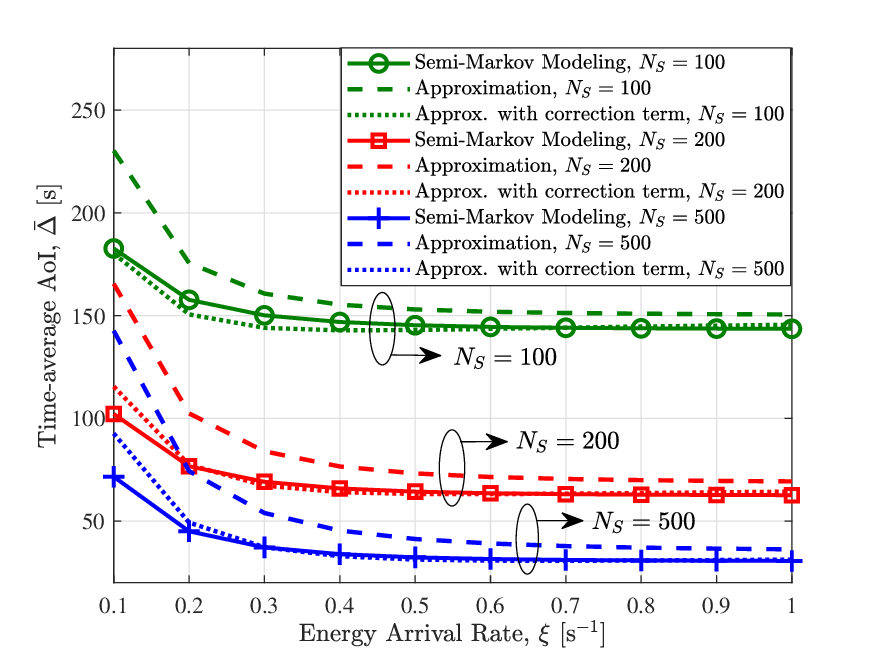}}
  \caption{Time-average AoI versus energy arrival rate. $\mu=0.05~\mathrm{packet/s}$, $N=10~\mathrm{energy~units/packet}$, $\theta=10\mathrm{dB}$.}
  \label{fig:AoIvsEHRate_combined}
    \vspace{-0.5cm}
\end{figure}
\begin{figure}[t] 
  \centering{}
    {\includegraphics[width=0.9\figwidth]{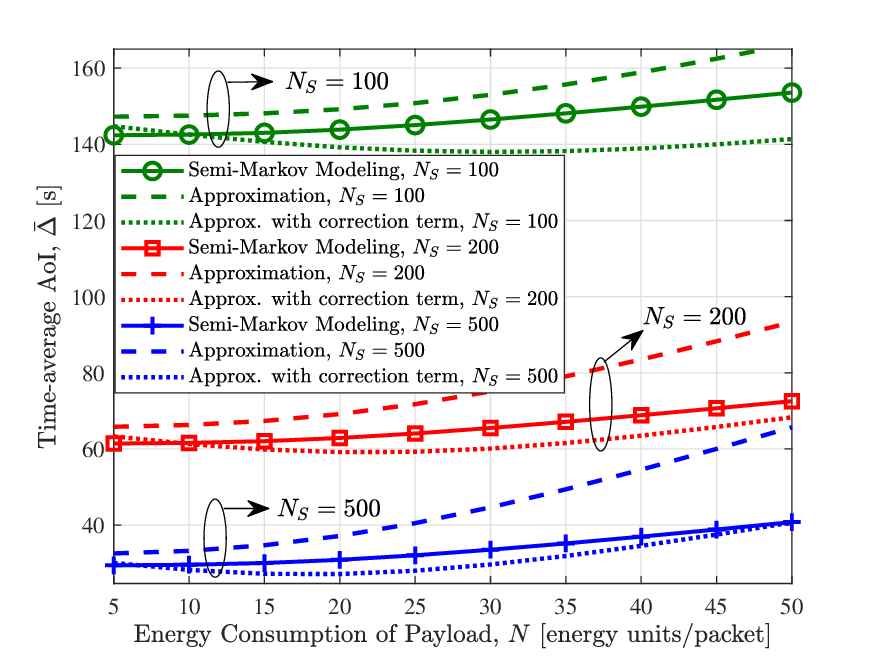}}
  \caption{Time-average AoI versus the energy consumption of payload.  $\xi=1~\mathrm{energy~units /s}$, $\mu=0.05~\mathrm{packet/s}$, $\theta=10\mathrm{dB}$.}
  \label{fig:11}
 \vspace{-0.5cm}
\end{figure}
~Fig.~\ref{fig:11} presents the time-average AoI as a function of the energy consumption for payload transmission. When $N$ is relatively small, the system can be regarded as operating in an energy-sufficient regime, where the time-average AoI is insensitive to the payload energy consumption. However, as payload energy consumption increases, the system enters the energy-constrained regime, leading to rapid degradation in the time-average AoI. Furthermore, a smaller total number of satellites results in a faster degradation of the time-average AoI, due to highly intermittent NTN conditions.

Fig.~\ref{fig:12} shows the time-average AoI as a function of the update rate of the source node. The result shows that the impact of $N_S$ depends on the energy regime. When NTNs operate in energy-sufficient regimes, the AoI performance shows little difference across satellite networks of varying densities. However, in energy-constrained scenarios, high-density satellite networks achieve significantly better AoI performance than low-density networks.

\begin{figure}[t] 
  \centering{}
    {\includegraphics[width=0.9\figwidth]{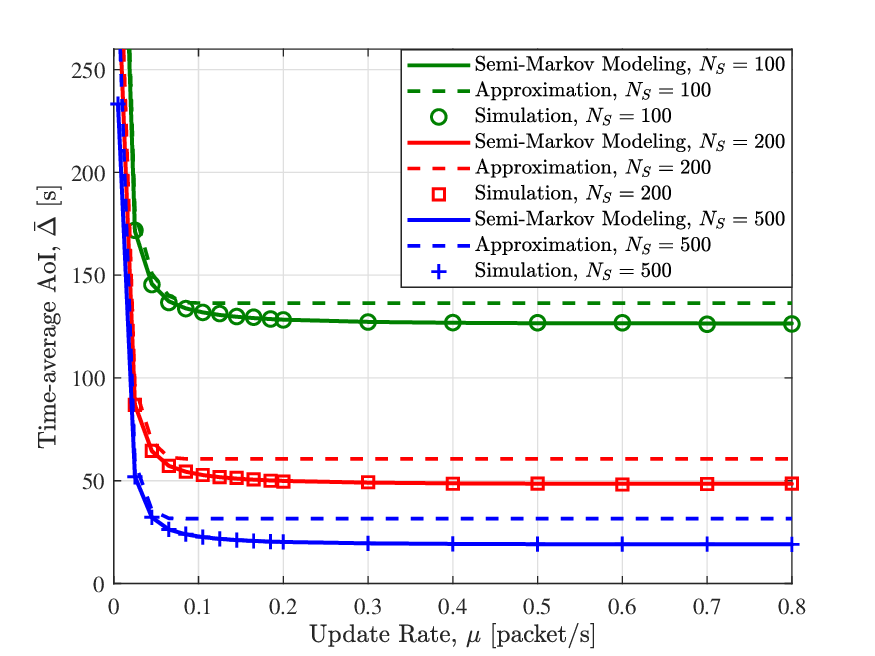}}
  \caption{Time-average AoI versus the update rate.  $\xi=0.5~\mathrm{energy~units /s}$, $N=10~\mathrm{energy~units/packet}$, $\theta=10\mathrm{dB}$.}
  \label{fig:12}
    \vspace{-0.5cm}
\end{figure}


Consider Fig.~\ref{fig:AoIvsheight}-Fig.~\ref{fig:12} together, replacing the detailed on-off transition dynamics with the average steady-state on probability provides a conservative upper bound for the accurate time-average AoI. Moreover, the gap between the two results becomes smaller when $P_\mathrm{on}$ is relatively small, such as when $N_s$ is small, or the SNR decoding threshold $\theta$ is high, and when the system has relatively sufficient energy supply, such as when $\mu$ and $N$ are small or $\xi$ is large. These observations indicate that the approximation is not merely a numerical simplification. Rather, it reveals that the detailed energy-channel coupling is essential. The approximation becomes accurate when either the channel availability is the dominant bottleneck or the energy constraint is inactive. Conversely, a large gap indicates a regime where the temporal synchronization between energy replenishment and on periods plays a significant role, and the semi-Markov model is necessary.

Fig.~\ref{fig:13} displays the time-average AoI comparison between the proposed probe strategy and a baseline approach that transmits directly without an initial probe, quantifying the energy-saving benefits of the probing mechanism. It is observed that adopting a probing strategy using a single energy unit to sense the channel is particularly advantageous when the satellite connectivity is highly intermittent, such as in cases of low satellite density or high SNR decoding thresholds. This method not only reduces power consumption but also yields significant improvements in information timeliness, as the harvested energy must accumulate over time before transmission, making the probing-based approach more efficient than direct transmission when the knowledge of the on-off connectivity status of the satellite is unknown.

\begin{figure}[t] 
  \centering{}
    {\includegraphics[width=0.9\figwidth]{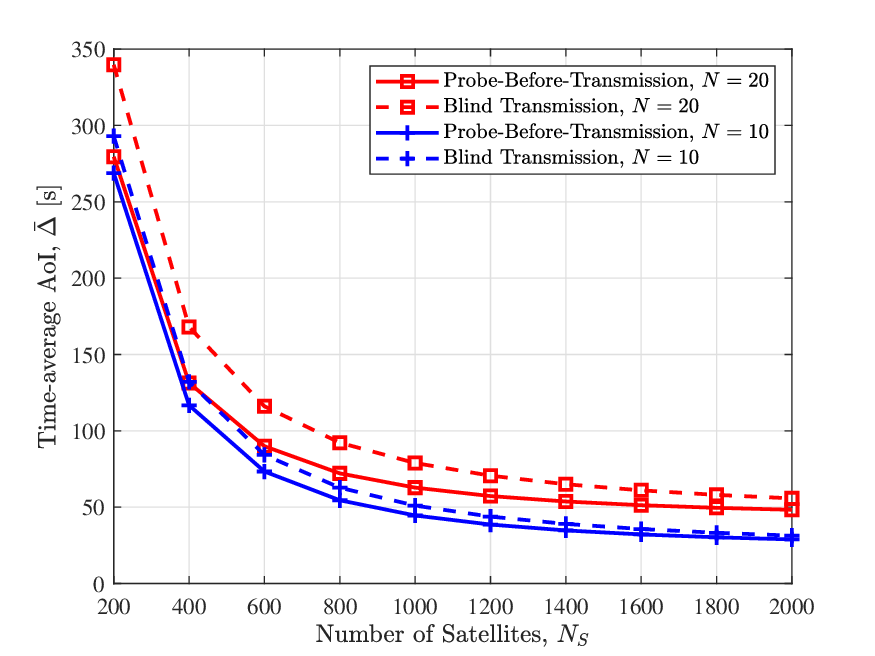}}
  \caption{Scheme comparison. $\mu=0.2~\mathrm{packet/s}$, $\xi=0.5~\mathrm{energy~units /s}$, $P_\mathrm{tx}=25~\mathrm{dBm}$, $\theta=10\mathrm{dB}$.}
  \label{fig:13}
    \vspace{-0.5cm}
\end{figure}

\color{black}

\section{Conclusion}

In this paper, we investigated the AoI in LEO satellite-assisted EH-IoT networks, where status updating is jointly constrained by intermittent satellite visibility and random energy availability. By modeling the LEO satellite deployment through a Poisson point process, we characterized the geometry-induced on/off contact process of the LEO satellite network. Building on this contact model, we developed a semi-Markov analytical framework that captures the coupling between the satellite channel state and the dynamic energy state under the probing scheme. This framework enables the characterization of the stationary energy-channel distribution and the derivation of the time-average AoI. Our results show that satellite intermittency and energy state jointly determine the achievable AoI. In particular, when satellite contacts are highly intermittent or the harvested energy is limited, a probing-based update mechanism can better exploit available satellite contact opportunities by avoiding unnecessary transmissions during unavailable periods, thereby improving AoI under such operating regimes. 

Future work may extend the proposed framework from AoI analysis to optimal status update control under the joint evolution of AoI, channel state, and energy state. It would also be valuable to incorporate more practical NTN-IoT factors, such as fading, interference, and random access contention, for performance analysis in large-scale satellite IoT networks.

\appendices
\section{}\label{proof:ONOFF} 

The off-service period $T_\mathrm{{off}}$ is the time elapsed until the next satellite enters the coverage dome of dome. We derive its distribution by mapping the spatial PPP to a temporal arrival process. This mapping is enabled by the relative motion (angular velocity $\omega$) between the user's dome and the static satellite distribution. Consider a small time interval $\epsilon$. Due to this motion, the dome scans a new region of area $A_\epsilon$, which is given by: $A_\epsilon = 2\left(R_E+h\right)^2\,\omega\,\epsilon\,\sin (\varphi_{ \mathrm{e}})$. By the fundamental properties of a homogeneous PPP:

$\bullet$ Independent Increments: The number of satellites in any disjoint area (e.g., areas scanned in disjoint time intervals) are independent random variables.

$\bullet$ Stationary Increments: The number of satellites $N(A_\epsilon)$ in the scanned area $A_\epsilon$ follows a Poisson distribution with mean $\mathbb{E}[N(A_\epsilon)] = \lambda A_\epsilon$. Since $A_\epsilon$ is directly proportional to $\epsilon$ (i.e., $A_\epsilon \propto \epsilon$), the resulting temporal satellite arrival process has stationary increments.

Since the satellite arrival is a counting process with stationary and independent increments, it constitutes a homogeneous Poisson process. The inter-arrival time of this process, $T_\mathrm{{off}}$, is therefore exponentially distributed. The rate of this Poisson process, $\lambda_{\mathrm{os}}$, is the expected number of arrivals per unit time:
\begin{equation}
\lambda_{\mathrm{os}} = \lim_{\epsilon \to 0} \frac{\mathbb{E}[N(A_\epsilon)]}{\epsilon} = \lim_{\epsilon \to 0} \frac{\lambda A_\epsilon}{\epsilon} = 2 \omega\lambda\sin (\varphi_{ \mathrm{e} }) \left(R_E+h\right)^2.
\end{equation}
Therefore, $T_{\mathrm{off}} \sim \exp(\lambda_{\mathrm{os}})$. As such, we can conclude that the off-service periods are exponentially distributed with the rate $2 \omega\lambda\sin (\varphi_{ \mathrm{e} }) \left(R_E+h\right)^2$.

On the other hand, the trajectory of each satellite in the dome section is characterized by the latitude difference between the source node and the point of the satellite projected on the ground, where these together determine the total service time of each satellite when establishing a connection with the source node. 
Let $\Theta$ denote the angle difference that represents the possible entry location of satellites, following a uniform distribution on [$- \varphi_{ \mathrm{e} }$, $ \varphi_{ \mathrm{e} }$]. 
The service time $T_\mathrm{{on}}$ is thus a random variable given by
\begin{align} 
    T_\mathrm{{on}} = \frac{2}{\omega}\arcsin\left( \tfrac{\sqrt{\sin^2 (\varphi_{\mathrm{e}})-\sin^2(\Theta)}}{\cos(\Theta)} \right).
\end{align}
Then, the CDF of $T_\mathrm{{on}}$ can be expressed as
\begin{equation} 
\begin{aligned}
    &F_{T_\mathrm{on}}(t)  = \mathbb{P}(T_\mathrm{on} < t)= \mathbb{P}\left(|\Theta| > \arcsin\sqrt{\tfrac{\sin^2\varphi_\mathrm{e} - \sin^2(\frac{\omega t}{2})}{1 - \sin^2(\frac{\omega t}{2})}}\right)\\
    & = 
    \begin{cases}
    1 - \frac{1}{\varphi_\mathrm{e}} \arcsin\sqrt{\frac{\sin^2\varphi_\mathrm{e} - \sin^2(\frac{\omega t}{2})}{1 - \sin^2(\frac{\omega t}{2})}}, & \text{if } t \in \left[0,\frac{2 \varphi_\mathrm{e}}{\omega}\right],  \\
    1, & \text{otherwise}.
    \end{cases}
\end{aligned}
\end{equation} 
The PDF of $T_\mathrm{{on}}$ can be derived using the distribution of $\Theta$.

\section{}\label{eq:AoI:matrix}
We derive the waiting-time equations for
$u_e$ and $v_e(a)$ by applying first-passage-time method. For $u_e$, we have
\begin{equation}\label{eq:FirstStep}
\begin{split}
     u_e=&dt+(\xi dt) u_{e+1}+ (\mu dt) u_{e-1}+(\lambda_\mathrm{os}dt) v_e(0)
     \\&~~~+(1-(\xi+\mu+\lambda_\mathrm{os})dt)u_e+o(dt),   
\end{split}
\end{equation}
where  $dt$ is an infinitesimal interval. The meaning of \eqref{eq:FirstStep} is that the current expected waiting time equals the first-step time plus the weighted average of the remaining expected waiting times after transitioning to the next state. Dividing both sides by  $dt$ and letting $dt\to 0$, \eqref{eq:FirstStep} can be transformed as
\begin{equation}
\xi\left(u_{e+1}-u_e\right)\!+\!\mu\left(u_{e-1}-u_e\right)\!+\!\lambda_\mathrm{os}\left(v_e(0)-u_e\right)=-1.
\end{equation}
In vector form, the off-state equation is
\begin{equation}
\mathbf{Q}_0\mathbf{u}+\lambda_\mathrm{os}(\mathbf v(0)-\mathbf u)=-\mathbf 1,
\end{equation}
where the off period waiting time vector $\mathbf{u}=[u_0,u_1,...,u_B]^{T}$, and the vector $\mathbf{v}(0)=[v_0(0),v_1(0),...,v_B(0)]^{T}$ represents the waiting time at the beginning of on period.


For the on period state, by treating a successful update as an absorbing event,
we use the killed generator $\mathbf{C}$ defined in \eqref{eq:C:define}. Using the scaled function
$\mathbf w_1(a)=\bar F(a)\mathbf v(a)$, we obtain
\begin{equation}\label{eq:w1}
\mathbf w_1(a)=
\int_a^{T_{\max}}e^{\mathbf{C}(t-a)}\left(\bar F(t)\mathbf 1+f(t)\mathbf u\right)dt.
\end{equation}

At the beginning of an on period, i.e., when $a=0$, the scaled factor satisfies $\bar F(a)=1$. Therefore, $\mathbf v(0)=\mathbf w_1(0)$. We define
$$
\mathbf{A}=\int_0^{T_{\max}}e^{\mathbf{C}t}f(t)\,dt,
\qquad
\mathbf{b}=\int_0^{T_{\max}}e^{\mathbf{C}t}\bar F(t)\mathbf 1\,dt.
$$
Substituting $\mathbf v(0)=\mathbf{b}+\mathbf{A}\mathbf{u}$ into the off-state
equation gives
\begin{equation}\label{eq:u}
\mathbf{u}=-(\mathbf{Q}_0-\lambda_{\mathrm{os}} \mathbf{I}+\lambda_{ \mathrm{os}} \mathbf{A})^{-1}
(\mathbf 1+\lambda_{\mathrm{os}}\mathbf b).
\end{equation}

Therefore, using
$\mathbf S^1(a)=\lambda_{ \mathrm{os}}\mathbf S^0e^{\mathbf{Q}_1a}\bar F(a)$, the time-average AoI expression can be written as Theorem \ref{Th:AoI:semi}. 

\section{}\label{eq:CTMC}
For ease of exposition, we denote by $\zeta_o=\mu P_{\mathrm{on}}$, and $\zeta_f=\mu (1-P_{\mathrm{on}})$. Then, the local balanced equations of the steady state of the CTMC can be derived by
\begin{subequations}\label{eq:finite_buffer_balance}
\begin{align}
&-\xi S_0
+\zeta_o S_{N+1}
=0,
\label{eq:infinite_buffer_balance_0}
\\
&-\xi S_i
+\xi S_{i-1}
+\zeta_o S_{i+N+1}
=0,~~1\le i\le N-1,
\label{eq:infinite_buffer_balance_low}
\\
&-\xi S_N
+\xi S_{N-1}
+\zeta_f S_{N+1}
+\zeta_o S_{2N+1}
=0,
\label{eq:infinite_buffer_balance_N}
\\
&-(\xi+\mu)S_i
+\xi S_{i-1}
+\zeta_f S_{i+1}+\zeta_o S_{i+N+1}
=0,
\nonumber\\
&\hspace{3.8cm}~N+1\leq i \leq B-N-1,
\label{eq:infinite_buffer_balance_mid}
\\
&-(\xi+\mu)S_i+\xi S_{i-1}+\zeta_fS_{i+1}=0,  \label{eq:infinite_buffer_balance_B1} \nonumber\\
&\hspace{3.8cm}~B-N-2\leq i \leq B-1,\\
&-\mu S_B+\xi S_{B-1}=0  \label{eq:BufferBound}.
\end{align}
\end{subequations}
\color{black}
Using the equation \eqref{eq:infinite_buffer_balance_mid} of $S_i$ when $i\geq N+1$, we approximate 
\begin{equation}\label{eq:SiN+1}
    S_i=\sum_{l=0}^{N+2} \bar{C}_{l}z_{l}^i = \bar{C}_d z_d^{i}+\sum_{l\neq d} \bar{C}_{l}z_{l}^i \approx C_d z_d^{i},
\end{equation}
where $z_l$ denotes all the roots of the following equation
\begin{equation}\label{eq:z:root:proof}
    \zeta_o z^{N+2}+\zeta_f z^2-(\xi+\mu)z+\xi =0,
\end{equation}
in the complex domain, $\bar{C}_i$ is the coefficient corresponding to the root. The approximation sign indicates that, to obtain a closed-form expression, we retain the dominant root $z_d$ as an approximation, where $C_d$ is the corresponding effective dominant-mode coefficient.

It can be proved that \eqref{eq:z:root:proof} only has one non-trivial positive real root ($z=1$ is always a trivial root). This root determines whether the energy distribution decays toward lower energy states or grows toward the high-energy boundary, and it is the dominant root\footnote{From a physical perspective, it is also reasonable to retain only this root. Negative real roots lead to sign-alternating terms, while complex roots introduce oscillatory components, neither of which is suitable for representing the dominant smooth trend of a steady-state probability distribution.}. 
Then, by using equation \eqref{eq:infinite_buffer_balance_0}, we have
\begin{equation}\label{eq:SiN+1:C}
    C_d=\frac{\xi S_0}{\zeta_o z^{N+1}}.
\end{equation}
Using the equation of $S_i$ when $1\leq i\leq N-1$, 
\begin{equation}
    S_i-S_{i-1}= S_0 z^i.
\end{equation}
Therefore, for $0\leq i\leq N-1$, we have
\begin{equation}\label{eq:Si0N-1}
    S_i=S_0\sum_{j=0}^{i}z^j \overset{(a)}=\frac{S_0(1-z^{i+1})}{1-z}. 
\end{equation}
Step ($a$) is derived by applying applying the geometric series summation when $z\neq 1$. When the dominant root satisfies $z=1$, the limiting value $S_i=iS_0$ should be taken.

Then, using \eqref{eq:infinite_buffer_balance_N}, the steady state probability $S_{N}$ can be derived as
\begin{equation}
    S_N=\left(\frac{1-z^{N+1}}{1-z}+\frac{\zeta_f}{\zeta_o}\right)S_0.
\end{equation}
For $B-N \leq i\leq B-1$, since from boundary condition \eqref{eq:BufferBound}, we then prove the steady state probability in this regime can be expressed as
\begin{equation}\label{eq:SBminM}
 S_{B-m}=S_B\Psi_m,~~~~m=0,1,...,N.
\end{equation}
Define the backward-indexed sequence $T_m=S_{B-m},~m=0,1,...,N.$
Then, the \eqref{eq:infinite_buffer_balance_B1} can be given by
\begin{equation}
   \xi T_{m+1}-(\xi+\mu)T_{m}+\zeta_f T_{m-1}=0.
\end{equation}
This is a second-order homogeneous linear difference equation with constant coefficients. To solve it, assume a solution of the form $T_m=r_m$. Then, we have the characteristic equation 
\begin{equation}
    \xi r^2- (\xi+\mu)r+\zeta_f=0. 
\end{equation}
Thus the two characteristic roots are \eqref{eq:r1} and  \eqref{eq:r2}, when $r_1\neq r_2$, the general solution is given by $T_m =Ar_1^m+Br_2^m$.
From $T_0=S_B$, and $T_1=S_{B-1}=\frac{\xi}{\mu}S_B$, we have 
\begin{equation}
\begin{cases}
     A+B=S_B \\
     Ar_1+Br_2=\frac{\mu}{\xi}S_B.
\end{cases}
\end{equation}
Therefore
\begin{equation}
    A=S_B\tfrac{\frac{\mu}{\xi}-r_2}{r_1-r_2},~~B=S_B\tfrac{r_1-\frac{\mu}{\xi}}{r_1-r_2}.
\end{equation}
Therefore, \eqref{eq:SBminM} can be proved, and 
\begin{equation}
    \Psi_m=\tfrac{
\left(\frac{\mu}{\xi}-r_2\right)r_1^m+\left(r_1-\frac{\mu}{\xi}\right)r_2^m}{r_1-r_2}.
\end{equation}
Equivalently, $S_i=S_B \Psi_{B-i},~B-N\leq i\leq B$. We assume $S_B=\Gamma_B S_0$. For state $i=B-N-1$, the equation \eqref{eq:infinite_buffer_balance_mid} can be given by
\begin{equation}
    -(\xi+\mu)S_{B-N-1}+\xi S_{B-N-2}+\zeta_f S_{B-N}+\zeta_o S_{B}=0.
\end{equation}
We can solve  $\Gamma_B$ as \eqref{eq:GammaB}.
Then, the steady state $S_0$ can be solved by using normalizing condition $\sum S_i=1$.

\bibliographystyle{IEEEtran}
\bibliography{abrv, ref}

\end{document}